\documentclass[preprint,12pt,3p,onecolumn]{elsarticle}

\usepackage{subfig}
\usepackage{graphicx}

\usepackage{dcolumn}
\usepackage{bm}
\usepackage{epstopdf}
\usepackage{epsfig}
\usepackage[colorlinks = true,
linkcolor = blue,
urlcolor  = blue,
citecolor = blue,
anchorcolor = blue]{hyperref}
\usepackage{url}
\usepackage{textcomp}

\usepackage{amsmath}
\usepackage{amssymb}
\usepackage{threeparttable,multirow}
\usepackage{amsfonts}
\usepackage{mathrsfs}
\usepackage{graphicx}
\usepackage{hyperref}
\usepackage{amsthm}
\usepackage{xcolor}
\newtheorem{theorem}{Theorem}
\newtheorem{lemma}{Lemma}

\begin{document}

\begin{frontmatter}

\title{Justification 
	of the discrete nonlinear Schr\"odinger equation from a parametrically driven damped nonlinear Klein-Gordon equation and numerical comparisons }

\author[label1,label5]{Y.\ Muda}
\address[label1]{Department of Mathematical Sciences,University of Essex, Colchester, CO4 3SQ, United Kingdom}
\address[label5]{Department of Mathematics, Faculty of Science and Technology,\\ 
	State Islamic University of Sultan Syarif Kasim Riau, Pekanbaru, 28294,  Indonesia\fnref{label4}}
\address[label2]{Theoretical Physics Laboratory,Theoretical High Energy Physics and Instrumentation Research Group, Faculty of Mathematics and Natural Sciences,  Institut Teknologi Bandung,  Bandung, Indonesia, 40132\fnref{label4}}
\address[label6]{Centre of Mathematical Modelling and Simulation, Institut Teknologi Bandung, 1st Floor, Labtek III,\\ Jl.\ Ganesha No.\ 10, Bandung, 40132, Indonesia}

\cortext[cor1]{Corresponding author}

\ead{ymuda@essex.ac.uk}

\author[label2]{F.T.\ Akbar}
\ead{ftakbar@fi.itb.ac.id}

\author[label1,label6]{R.\ Kusdiantara}
\ead{rkusdi@essex.ac.uk}

\author[label2]{B.E.\ Gunara}
\ead{bobby@fi.itb.ac.id}

\author[label1]{H.\ Susanto\corref{cor1}}
\ead{hsusanto@essex.ac.uk}

\begin{abstract}
We consider a damped, parametrically driven discrete nonlinear Klein-Gordon equation, that models coupled pendula and micromechanical arrays, among others. To study the equation, one usually uses a small-amplitude wave ansatz, that reduces the equation into a discrete nonlinear Schr\"odinger equation with damping and parametric drive. Here, we justify the approximation by looking for the error bound with the method of energy estimates. Furthermore,  we prove the local and global existence of {solutions to the discrete nonlinear} Schr\"odinger equation.  To illustrate the main results, we consider numerical simulations showing the dynamics of errors made by the discrete nonlinear equation. We consider two types of initial conditions, with one of them being a discrete soliton of the nonlinear Schr\"odinger equation, that is expectedly approximate discrete breathers of the nonlinear Klein-Gordon equation.
\end{abstract}


\end{frontmatter}


\section{Introduction}\label{intro}

We consider 
the following parametrically driven discrete Klein-Gordon (dKG) equation with damping 
\begin{equation}\label{Parametric}
	\ddot{u}_j = -u_j - \xi u_j^3 + \varepsilon^2\Delta u_j - \alpha \dot{u}_j + H\cos\left(2\Omega t\right)
	u_j,
\end{equation}
{where $u_j \equiv u_j(t)$ is a real-valued wave function at site $j$, the overdot {denotes} the time derivative and $\varepsilon^2$ represents the coupling constant between two adjacent sites, with $\Delta u_j = u_{j+1} - 2u_j + u_{j-1}$ being the {one dimensional} discrete Laplacian. The {positive} parameters $\alpha$ and $H$ denote the damping coefficient and the strength of the parametric drive, respectively. The real constant $\xi$ is the nonlinearity coefficient and $\Omega$ is the driving frequency. 
The governing equation \eqref{Parametric} is relevant to the experimental study of localised structures in coupled pendula \cite{dena92,cuev09} and micromechanical arrays \cite{buks02}. 
	
To analyse the equation, one usually uses a multiple scale expansion and the rotating wave approximation under the assumption of small wave amplitudes, that lead to a damped, parametrically driven discrete nonlinear Schr\"odinger (dNLS) equation \cite{dena92b,putt93,lifs03,morgante2002standing}. Using the scaling $\alpha = \varepsilon^2\hat{\alpha}$, $H= 2\varepsilon^2h$, $\Omega = 1+{\varepsilon^2\Lambda}/{2}$, and a slow time variable $\tau = {\varepsilon^2 t}/{2}$, we consider a $(2:1)$ parametric resonance and define a slowly varying approximation to the solutions of the dKG lattice equation \eqref{Parametric} 
\begin{equation} \label{ansatzP}
u_j(t) \approx \phi_j(t) = \varepsilon\;A_j(\tau) e^{i\Omega t} + \frac{\varepsilon^3}{8}\left[\xi A_j(\tau)^3 - h A_j(\tau)\right] e^{3i\Omega t} + \mathrm{c.c.},
\end{equation}
that will yield the dNLS equation 
	\begin{equation}\label{dNLS}
		i \dot{A}_j = \Delta A_j  - i \hat{\alpha} 
		A_j + \Lambda  A_j - 3 \xi  |A_j|^2 A_j + h \bar{A}_j\:.
	\end{equation}
Here, the dot denotes derivative with respect to the slow time $\tau$, which implies that the approximation \eqref{ansatzP}-\eqref{dNLS} is expected to be valid until $t\sim\mathcal{O}(2/\varepsilon^2)$. The abbreviation $\mathrm{c.c.}$ means the complex conjugate of the preceding terms. It should be clear by now that the coupling constant (i.e., the prefactor of the discrete Laplacian term) is scaled to $\varepsilon^2$ only for the sake of convenience, so that $u_j = \mathcal{O}(\varepsilon)$. Replacing $\varepsilon\to\sqrt{\varepsilon}$ will yield the standard scaling used, e.g., in \cite{pelinovsky2016approximation}. Our scaling may also be interpreted that instead of using the coupling constant as a measure of the smallness, we use the solution amplitude.

The presence of parametric drive and damping in the dNLS equation was possibly first studied in \cite{hennig1999periodic}, where the existence of localised solutions was discussed using a nonlinear map approach. It was shown that numerous types of localised states emerge from the system depending on the strength of the parametric driving. The parametrically driven dNLS equation (\ref{dNLS}) was studied in \cite{susanto2006stability,syafwan2010discrete}, where it was shown that the parametric drive can change the stability of fundamental discrete solitons, i.e., it can destroy onsite solitons as well as restore the stability of intersite discrete solitons, both for bright and dark cases. In \cite{syaf12}, breathers of \eqref{dNLS}, i.e., spatially localised solutions with periodically time varying $|A_j(\tau)|$ emanating from Hopf bifurcations, were studied systematically. 
	
Despite the wide interests in both equations \eqref{Parametric} and \eqref{dNLS}, the reduction from the former to the latter has not been rigorously justified. Without damping and parametric drive, the analysis was provided rather recently by Pelinovsky, Penati and Paleari \cite{pelinovsky2016approximation}. However, as the presence of damping and drive will certainly 
require modifications in the justification of the 
reduction, here we address the problem, which will be the primary aim of the paper. 
Following \cite{pelinovsky2016approximation}, we use an energy estimate method. The method has been used as well in various systems of differential equations, see, e.g., \cite{bamb06,bide13,duma14,frie99,schn99}.

The paper is organized as follows. Mathematical formulations to obtain the uniqueness and global existence of solutions to the dNLS equation \eqref{dNLS} are given in Section \ref{sec:1}. 
In Section \ref{sec3}, we discuss an error bound estimation of the rotating wave approximation, which leads to the main result of the paper, i.e., Theorem \ref{theorem1}. Finally, in Section \ref{sec4} we illustrate the main results by considering the evolution of errors made by the rotating wave approximation for two different initial conditions, with one of them corresponding to 
discrete solitons of the nonlinear Schr\"odinger equation \eqref{dNLS}. 
	
\section{Analytical formulation and preliminary results}\label{sec:1}

{Substituting the slowly varying approximation ansatz \eqref{ansatzP} into the original dKG equation \eqref{Parametric} and taking into account the dNLS equation \eqref{dNLS}, we obtain the residual terms in the form of}
\begin{equation}
\begin{aligned}
&R_j(t) :=\nonumber\\
& \varepsilon^5 \left[ \frac{e^{i\Omega t}}{8} \left( -3 h \xi  A_j \bar{A}_j^2 + 3 \xi ^2 A_j^3 \bar{A}_j^2 + h^2 A_j - h \xi  A_j^3 + 4 i \alpha  \Lambda  A_j + 4 \alpha  \dot{A}_j - 2\Lambda ^2 A_j + 4 i \Lambda  \dot{A}_j + 2\ddot{A}_j \right)\right. \nonumber \\ 
& \quad \left.+ \frac{e^{3i\Omega t}}{8} \left( -6 h \xi  A_j^2 \bar{A}_j + 6 \xi ^2 A_j^4 \bar{A}_j - 3 i \alpha  h A_j + 9 h \Lambda  A_j - 3 i h \dot{A}_j - 2 h A_j + h A_{j-1}+\frac{1}{8} h A_{j+1} \right.\right. \nonumber \\
& \quad \left.\left.+ 3 i \alpha  \xi  A_j^3 - 9 \xi  \Lambda  A_j^3 + 9 i \xi  A_j^2 \dot{A}_j + 2\xi  A_j^3 - \xi  A_{j-1}^3 - \xi  A_{j+1}^3 \right.\bigg)\right.\nonumber\\
&\quad\left.+ \frac{e^{5i\Omega t}}{8} \left( h^2 A_j - 4 h \xi  A_j^3 + 3 \xi ^2 A_j^5 \right) \right]  \nonumber\\
& + \varepsilon^7 \left[ \frac{3e^{i\Omega t}}{32}  \left(  h^2 \xi  A_j^2 \bar{A}_j - h \xi ^2 A_j^4 \bar{A}_j - h \xi ^2 A_j^2 \bar{A}_j^3 +  \xi ^3 A_j^4 \bar{A}_j^3 \right) +  \frac{e^{3i\Omega t}}{32} \left( -6 i \alpha  h \Lambda  A_j - 2 \alpha  h \dot{A}_j - h \ddot{A}_j\right. \right. \nonumber \\
& \quad \left. \left.+ 9 h \Lambda ^2 A_j - 6 i h \Lambda  \dot{A}_j + 6 i \alpha  \xi  \Lambda  A_j^3 + 6 \alpha  \xi  A_j^2 \dot{A}_j - 9 \xi  \Lambda ^2 A_j^3 + 18 i \xi  \Lambda  A_j^2 \dot{A}_j + 6 \xi  A_j \dot{A}_j^2 + 3 \xi  A_j^2 \ddot{A}_j \right) \right. \nonumber \\
& \quad \left.+ \frac{e^{5i\Omega t}}{64} \left( 3 h^2 \xi  A_j^2 \bar{A}_j - 6 h \xi ^2 A_j^4 \bar{A}_j + 3 \xi ^3 A_j^6 \bar{A}_j \right) + \frac{e^{7i\Omega t}}{64} \left( 3 h^2 \xi  A_j^3 - 6 h \xi ^2 A_j^5 + 3 \xi ^3 A_j^7 \right) \right] \nonumber\\
& + \varepsilon^9 \left[ \frac{e^{3i\Omega t}}{512} \left( -3 h^3 \xi  A_j^2 \bar{A}_j+6 h^2 \xi ^2 A_j^4 \bar{A}_j+6 h^2 \xi ^2 A_j^2 \bar{A}_j^3-3 h \xi ^3 A_j^6 \bar{A}_j-6 h \xi ^3 A_j^4 \bar{A}_j^3+3 \xi ^4 A_j^6 \bar{A}_j^3 \right)\right.\nonumber\\
&\quad \left.+ \frac{e^{9i\Omega t}}{512}  \left( - h^3 \xi  A_j^3 + 3 h^2 \xi ^2 A_j^5 - 3 h \xi ^3 A_j^7 + \xi ^4 A_j^9 \right) \right] + \mathrm{c.c.}
\end{aligned}
\label{error}
\end{equation}

Note that the dNLS equation \eqref{dNLS} is obtained from removing the resonant terms at $\mathcal{O}(\varepsilon^3)$. The usual rotating frame ansatz, where one uses only the first term of \eqref{ansatzP} and its complex conjugate, i.e., $u_j(t) \approx \varepsilon\;A_j(\tau) e^{i\Omega t} + \mathrm{c.c.},$ will also yield \eqref{dNLS}, but it leaves a larger residue: 
\begin{equation*} 
\begin{array}{ccl} 
R_j(t)  &:= & \varepsilon^3 (\xi A_j^3 - hA_j) e^{3i\Omega t} + \varepsilon^5 \left(\frac{1}{2} i \hat{\alpha} \Lambda A_j   - \frac{1}{4}\Lambda^2 A_j + \frac{1}{2} \hat{\alpha} \dot{A}_j + \frac{1}{2} i \Lambda \dot{A}_j + \frac{1}{2} \ddot{A}_j \right) e^{i\Omega t} +c.c.
\end{array}
\end{equation*} 
Over a long time, this can yield an error that is of the same order as the approximation itself, which will be problematic for the expansion. Thus, $R_j$ needs to be smaller. It might be made small only if $\xi$ and $h$ were both to be very small. However, this cannot be the case since they have been derived from scaling the original "physical" parameters properly with respect to $\varepsilon$, i.e., their smallness have already been exploited. The only way to reduce the residue is therefore by modifying the rotating ansatz, which following \cite{dorf11} (see Chapter 5) yields \eqref{ansatzP}. 

In the followings, we denote by $A$ the sequence $(A_j)_{j\in\mathbb{Z}}$ in $\ell^2(\mathbb{Z})$, which 
is a Banach space equipped with norm,
\begin{equation}
\|A\|_{\ell^2(\mathbb{Z})} = \left(\sum_{j\in\mathbb{Z}} \;|A_j|^2 \right)^{1/2}.
\end{equation}
First, we prove the preliminary estimates on the global solutions of the dNLS equation \eqref{dNLS} in $\ell^2(\mathbb{Z})$-space, the leading order approximation (\ref{ansatzP}), and the residual term (\ref{error}).

\begin{lemma}\label{Global}
	For every $A(0) = \varphi \in \ell^2 (\mathbb{Z})$, the dNLS equation \eqref{dNLS} admits a unique global solution $A(\tau)$ on $[0,\infty)$ which belongs to $C^k\left([0,+\infty), \ell^2 (\mathbb{Z}) \right)$. 
	Furthermore, the unique solution $A(\tau)$ satisfies the estimate  
	\begin{equation}
	\|A(\tau)\|_{\ell^2 (\mathbb{Z})} \leq \| \varphi \|_{\ell^{2}(\mathbb{Z})}\;e^{-(\hat{\alpha} - 2|h|)\tau}\:.
	\end{equation}
\end{lemma}
\begin{proof}
	We split the proof into four parts.\\
	
	\textbf{1. Local existence.} Let us rewrite Eq.\ \eqref{dNLS} in its equivalent integral form \\
	\begin{equation}\label{Intform}
	A_j(\tau) = \varphi_j - i \int_{0}^{\tau} \left(\Delta A_j - i \alpha A_j + \Lambda A_j - 3 \xi |A_j|^2 A_j  + h \bar{A}_j\right) ds.\,
	\end{equation}
	Define a Banach space,
	\begin{equation}
	\mathcal{B} = \{A \in  C\left([0,\tilde{\tau}], \ell^2 (\mathbb{Z}) \right) | \|A\|_{\ell^2 (\mathbb{Z})} \leq \delta\}\;,
	\end{equation}
	equipped with norm,
	\begin{equation}
	\|A\|_{\mathcal{B}} = \sup_{\tau \in [0,\tilde{\tau}]} \| A (\tau) \|_{\ell^2 (\mathbb{Z})}.
	\end{equation}
	
	For $A \in \ell^2 (\mathbb{Z})$, we define a nonlinear operator
	\begin{equation}\label{operatorK}
	K_j\left[A(\tau)\right] = \varphi_j - i \int_{0}^{\tau} \left(\Delta A_j - i \alpha A_j + \Lambda A_j - 3 \xi |A_j|^2 A_j  + h \bar{A}_j\right) ds.
	\end{equation}
	We want to prove that the operator $K$ is a contraction mapping on $\mathcal{B}$. Because the discrete Laplacian $\Delta$ is a bounded operator on $\ell^2(\mathbb{Z})$, we have 
	\begin{equation}
	\| \Delta A \|_{\ell^2(\mathbb{Z})} \leq C_{\Delta} \| A \|_{\ell^2(\mathbb{Z})}.
	\end{equation}
To be precise, $C_\Delta=4$ because the operator is a self-adjoint and its continuous spectrum lies within the interval $[-4,0]$. 

	Since $\ell^2(\mathbb{Z})$ is an algebra, there is a constant $C > 0$ such that for every $A, B \in \ell^2(\mathbb{Z})$, we have
	\begin{equation}\label{Banach}
	\| A B \|_{\ell^2(\mathbb{Z})} \leq C \| A \|_{\ell^2(\mathbb{Z})} \| B \|_{\ell^2(\mathbb{Z})}.
	\end{equation}
	From Eq.\ \eqref{operatorK} and using the estimate \eqref{Banach}, we obtain the following bound 
	\begin{equation}
	\| K (A)\|_{\mathcal{B}} \leq \delta_0 + \tilde{\tau}(C\delta + \hat{\alpha} \delta + \Lambda \delta+ 3 |\xi|\delta^3 + |h| \delta).
	\end{equation}
	We can pick $\delta_0 < \frac{\delta}{2}$ and $\tilde{\tau} \leq \frac{\delta}{2(C\delta + \hat{\alpha} \delta + \Lambda \delta+ 3 |\xi|\delta^3 + |h| \delta)}$ to conclude that $K : \mathcal{B} \rightarrow \mathcal{B}$.

	Let $A,B \in \mathcal{B}$. Then we have 
	\begin{eqnarray}
	K_j[A(\tau)] - K_j[B(\tau)] & = & -i \int_{0}^{\tau} \left[\Delta(A_j-B_j) - i\hat{\alpha} (A_j-B_j) + \Lambda (A_j-B_j)  \right. \nonumber\\
	& & \qquad \left.- 3 \xi (|A_j|^2A_j - |B_j|^2 B_j) + h (\bar{A}_j-\bar{B}_j)\right] \:ds\:.
	\end{eqnarray}
	Noting that 
	\begin{equation} \label{nonlinearProp}
	\begin{array}{cll}
	|A_j|^2 A_j - |B_j|^2 B_j &=& |A_j|^2 A_j -|A_j|^2 B_j+ |A_j|^2 B_j - |B_j|^2 B_j\\
	&=& |A_j|^2(A_j - B_j) + B_j(|A_j|^2-|B_j|^2)\\
	&=& |A_j|^2(A_j - B_j) + B_j(A_j\bar{A}_j - A_j\bar{B}_j + A_j\bar{B}_j - B_j\bar{B}_j)\\
	&=& |A_j|^2(A_j-B_j) + B_j[A_j(\bar{A}_j - \bar{B}_j) + (A_j - B_j)\bar{B}_j],
	\end{array}
	\end{equation}	
	we obtain that
	\begin{equation}
	\| K(A) - K(B) \|_{\mathcal{B}}\leq \tilde{\tau} \left(C_\Delta + \hat{\alpha} + \Lambda + 3 |\xi| C \delta^2 + |h|\right)\| A-B \|_{\mathcal{B}}.
	\end{equation}
	By taking 
	\begin{equation}
	\tilde{\tau} < \mathrm{min}\left(\frac{1}{C_\Delta + \hat{\alpha} + \Lambda + 3 |\xi| C \delta^2 + |h|}, \frac{\delta}{2(C\delta + \hat{\alpha} \delta + \Lambda \delta+ 3 |\xi|\delta^3 + |h| \delta)}\right) \:,
	\end{equation}
	then $K$ is a contraction mapping on $\mathcal{B}$. Therefore, by Banach fixed point theorem, there exists a unique fixed point of operator $K$, which is a unique solution of \eqref{Intform}. \\
	
	\textbf{2. Smoothness.} From Eq.\ \eqref{dNLS}, we obtain that
	\begin{equation}
	\sup_{\tau \in [0,\tilde{\tau}]} \|\dot{A}\|_{\ell^2(\mathbb{Z})} \leq \left(C_\Delta + \hat{\alpha} + \Lambda + |h| + 3|\xi|\delta^2 \right) \delta\:,
	\end{equation}
	which shows that the solution belongs to $C^1\left([0,\tilde{\tau}), \ell^2 (\mathbb{Z}) \right)$.
	
	Furthermore, by writing $\eta = (A, \bar{A})$, we have
	\begin{equation}\label{etaEq}
	i\frac{d\eta}{d\tau} = L\eta + N(\eta) + F(\eta)\:,
	\end{equation} 
	where
	\begin{equation}
	L\eta = \left[ \begin{matrix}
	\Delta - i\hat{\alpha} + \Lambda \\
	-\Delta - i\hat{\alpha} - \Lambda
	\end{matrix}\right]\eta\:, \qquad N(\eta) = \frac{3\xi}{2} |\eta|^2 \eta\:, \qquad F(\eta) = \left[\begin{matrix}
	0,h\\h,0
	\end{matrix}\right]\eta\:.
	\end{equation}
	Differentiating \eqref{etaEq}, we obtain
	\begin{equation}
	i\frac{d^2\eta}{d\tau^2} = L\frac{d\eta}{d\tau} + DN(\eta)\cdot\frac{d\eta}{d\tau} + DF(\eta)\cdot\frac{d\eta}{d\tau}\:.
	\end{equation}
	Since $N(\eta)$ and $F(\eta)$ are smooth on $\ell^2 (\mathbb{Z})$ and $\frac{d\eta}{d\tau} \in C\left([0,\tilde{\tau}), \ell^2 (\mathbb{Z}) \right)$, then we have
	\begin{equation}
	\frac{d^2\eta}{d\tau^2} \in C\left([0,\tilde{\tau}), \ell^2 (\mathbb{Z}) \right)\:,
	\end{equation}
	which implies that $A \in C^2\left([0,\tilde{\tau}), \ell^2 (\mathbb{Z}) \right)$. Using a similar procedure for higher derivatives, we conclude that 
	\begin{equation}
	A \in C^k\left([0,\tilde{\tau}), \ell^2 (\mathbb{Z}) \right)\:.
	\end{equation}
	\\

	\textbf{3. Maximal solutions.} We can construct a maximal solution by repeating the steps above with the initial condition $A(\tilde{\tau} - \tau_0)$ for some $0<\tau_0 < \tilde{\tau}$ and by using the uniqueness result to glue the solutions.\\
	
	\textbf{4. Global existence.} To prove the global existence, first we multiply the dNLS equation \eqref{dNLS} with $\bar{A}_j$ and use its complex conjugate to obtain
	\begin{equation}\label{dNLS2}
	i\frac{d}{d\tau} \left(A_j \bar{A}_j\right) - \left(\bar{A}_j \Delta A_j - A_j \Delta \bar{A}\right)=  - 2 i\hat{\alpha} |A_j|^2 + h (\bar{A}_j^2 - A_j^2)\:.
	\end{equation}
	Note that
	\begin{equation}
	\sum_{j \in \mathbb{Z}} \left(A_j \Delta\bar{A}_j - \bar{A}_j \Delta A_j\right) = 0.
	\end{equation}
	Summing up \eqref{dNLS2} over $j$, we then get 
	\begin{equation}
	\frac{d}{d\tau} \|A\|^2_{\ell^{2}(\mathbb{Z})} =  - 2 \hat{\alpha} \|A\|^2_{\ell^{2}(\mathbb{Z})}  + 4 \; h \sum_{j \in \mathbb{Z}}\:\mathrm{Im}(\bar{A}_j)\mathrm{Re}(\bar{A}_j).
	\end{equation}
	For the last term in the above equation, we have the estimate
	\begin{equation}
	h \sum_{j \in \mathbb{Z}}\:\mathrm{Im}(\bar{A}_j)\mathrm{Re}(\bar{A}_j) \leq |h| \sum_{j \in \mathbb{Z}}\:\left|\mathrm{Im}(\bar{A}_j)\right| \left|\mathrm{Re}(\bar{A}_j)\right| \leq |h| \|A\|^2_{\ell^{2}_k(\mathbb{Z})}\:,
	\end{equation}
	which leads to
	\begin{equation}
	\dfrac{d}{d\tau} \| A \|^2_{\ell^{2}(\mathbb{Z})} + 2 (\hat{\alpha} - 2h)\| A \|^2_{\ell^{2}(\mathbb{Z})} \leq 0\:.
	\end{equation}
	Integrating the inequality, we get
	\begin{equation} \label{GlobalBoundl2}
	\| A(\tau) \|_{\ell^{2}(\mathbb{Z})} \leq \| \varphi \|_{\ell^{2}(\mathbb{Z})}\;e^{-(\hat{\alpha} - 2|h|)\tau},
	\end{equation}
	which shows that $A(\tau)$ cannot blow up in finite time. Thus, the dNLS equation \eqref{dNLS} admits global solutions. 
	
\end{proof}

It is worth mentioning that due to the damping term, the dNLS equation \eqref{dNLS} does not possess a constant of motion. However, we can define a Hamiltonian (i.e., an energy function) associated with equation \eqref{dNLS} as
\begin{equation}
\mathrm{H}_{\mathrm{dNLS}}[A](\tau) = \sum_{j \in \mathbb{Z}}\;\left(|\nabla A_j|^2 - \Lambda |A_j|^2 + \frac{3\xi}{2}|A_j|^4 - h \mathrm{Re}(A_j^2)\right),
\end{equation}
with $\nabla A_j = A_{j+1} - A_{j}$. The Hamiltonian function satisfies the differential equation,
\begin{equation}
\frac{d}{d\tau}\mathrm{H}_{\mathrm{dNLS}}[A](\tau) + 2\alpha \mathrm{H}_{\mathrm{dNLS}}[A](\tau) = - 3\alpha\xi \sum_{j \in \mathbb{Z}}\;|A_j(\tau)|^4\:.
\end{equation}
When $\alpha = 0$ (i.e., there is no damping present), $\mathrm{H}_{\mathrm{dNLS}}$ is conserved.

Now we provide estimates for the leading order approximation (\ref{ansatzP}) and the residual terms \eqref{error} in the following lemmas.

\begin{lemma}\label{ConstantX}
	For every $A(0) = \varphi \in \ell^2(\mathbb{Z})$, there exists a $\varepsilon$-independent positive constant $C_\phi$ that depends on $\lVert \varphi\rVert_{\ell^2(\mathbb{Z})}, h, \hat{\alpha}$ and $\tau_0$, such that the leading-order approximation (\ref{ansatzP}) satisfies
	\begin{equation}
	\lVert \phi(t)\rVert_{\ell^2(\mathbb{Z})} + \lVert \dot{\phi}(t) \rVert_{\ell^2(\mathbb{Z})} \leq \varepsilon \;C_\phi,
	\label{tam2}
	\end{equation}
	for all $t \in [0,2\tau_0/\varepsilon^2]$ and $\varepsilon \in (0,1)$.
\end{lemma}	

\begin{proof}
	From the global existence in Lemma \ref{Global} and using the Banach algebra property of $\ell^2(\mathbb{Z})$, we obtain 
	\begin{equation}\label{Phi}
	\begin{array}{ccl}
	\left\lVert \phi(t)\right \rVert_{\ell^2(\mathbb{Z})} & = &\left\lVert \varepsilon \left(A_j e^{i \Omega t} + \bar{A}_j e^{-i \Omega t}\right) + \frac{\varepsilon^3}{8} \left[\left(\xi A_j^3 - h A_j\right) e^{3 i\Omega t} \right.\right.\\
	&& \left.\left. \quad+ \left(\xi \bar{A}_j^3 - h \bar{A}_j\right)e^{-3 i\Omega t}\right]\right\rVert_{\ell^2(\mathbb{Z})}\\
	&\leq& \left\lVert\varepsilon \left(A_j e^{i \Omega t} + \bar{A}_j e^{-i \Omega t}\right) \right\rVert_{\ell^2(\mathbb{Z})} + \left\lVert \frac{\varepsilon^3}{8}\left(\xi A_j^3 e^{3 i\Omega t} -h A_j e^{3 i\Omega t}\right)\right\rVert_{\ell^2(\mathbb{Z})} \\
	&& \quad + \left\lVert \frac{\varepsilon^3}{8}\left( \xi \bar{A}_j^3 e^{-3 i\Omega t} - h \bar{A}_j e^{-3 i\Omega t}\right)\right\rVert_{\ell^2(\mathbb{Z})}\\
	& \leq& \varepsilon \left(2 \left\lVert A\right\rVert_{\ell^2(\mathbb{Z})} + \frac{1}{4}|\xi| \varepsilon^2 \left\lVert A^3\right\rVert_{\ell^2(\mathbb{Z})} + \frac{1}{4} |h| \varepsilon^2 \left\lVert A\right\rVert_{\ell^2(\mathbb{Z})} \right)  \\
	& \leq& \varepsilon \;C_{\phi_1}
	\end{array}
	\end{equation}
	and 
	\begin{equation}\label{Pdot}
	\begin{array}{ccl}
	\left\lVert\dot{\phi}(t)\right\rVert_{\ell^2(\mathbb{Z})} & = & \left\lVert\dfrac{3}{16} i h \Lambda  \varepsilon^5 e^{-3 i\Omega t} \bar{A}_j - \dfrac{1}{16} h \varepsilon^5 e^{-3 i\Omega t} \dot{\bar{A}}_j + \dfrac{3}{8} i h \varepsilon^3 e^{-3 i\Omega t} \bar{A}_j  \right. \\
	&& \quad \left. - \dfrac{3}{16} i \xi  \Lambda  \varepsilon ^5 e^{-3i\Omega t} \bar{A}_j^3 +\dfrac{3}{16} \xi  \varepsilon ^5 e^{-3 i\Omega t} \bar{A}_j^2 \dot{\bar{A}}_j - \dfrac{3}{8} i \xi  \varepsilon^3 e^{-3i\Omega t} \bar{A}_j^3  \right.\\
	&& \quad \left. - \dfrac{1}{2} i \Lambda  \varepsilon^3 e^{-i\Omega t} \bar{A}_j + \dfrac{1}{2} \varepsilon^3 e^{-i\Omega t} \dot{\bar{A}}_j -i \varepsilon  e^{-i\Omega t} \bar{A}_j \right. \\
	&& \quad \left. - \dfrac{3}{16} i h \Lambda  \varepsilon ^5 e^{3i\Omega t} A_j - \dfrac{1}{16} h \varepsilon^5 e^{3 i\Omega t} \dot{A}_j - \dfrac{3}{8} i h \varepsilon^3 e^{3i\Omega t} A_j  \right.\\
	&& \quad \left. +\dfrac{3}{16} i \xi \Lambda \varepsilon^5 e^{3i\Omega t} A_j^3 + \dfrac{3}{16} \xi  \varepsilon^5 e^{3 i \Omega t} A_j^2 \dot{A}_j + \dfrac{3}{8} i \xi  \varepsilon^3 e^{3i\Omega t} A_j^3  \right.\\
	&& \quad \left. + \dfrac{1}{2} i \Lambda  \varepsilon^3 e^{i\Omega t} A_j +\dfrac{1}{2} \varepsilon^3 e^{i\Omega t} \dot{A}_j+i \varepsilon  e^{i \Omega t} A_j \right\rVert_{\ell^2(\mathbb{Z})}.
	\end{array} 
	\end{equation}
	Since $A \in C^1\left([0,+\infty), \ell^2 (\mathbb{Z}) \right)$, then we have
	\begin{equation}\label{Pdot2}
	\left\lVert\dot{\phi}(t)\right\rVert_{\ell^2(\mathbb{Z})} \leq \varepsilon\;C_{\phi_2}\:.
	\end{equation}
	From Eqs.\ \eqref{Phi} and \eqref{Pdot2}, we obtain the inequality \eqref{tam2}, which concludes the proof.
\end{proof}

\begin{lemma}\label{ConstantR}
	For every $A(0) = \varphi \in \ell^2(\mathbb{Z})$, there exists a positive $\varepsilon-$independent constant $\tilde{C_R}$ that depends on $\lVert A_0\rVert_{\ell^2}, h, \hat{\alpha}$ and $\tau_0$, such that for every $\varepsilon \in (0,1)$ and every $ t \in [0,2\tau_0/\varepsilon^2]$, the residual terms in (\ref{error}) is estimated by
	\begin{equation}\label{Residu}
	\lVert R(t) \rVert_{\ell^2(\mathbb{Z})} \leq \tilde{C_R} \varepsilon^5.
	\end{equation}
\end{lemma}	

\begin{proof}
	To prove this lemma, we can use the result from Lemma \ref{Global} as well as the property of Banach algebra in $\ell^2(\mathbb{Z})$, such that from the global existence and smoothness of the solution $A(\tau)$ of the discrete nonlinear Schr\"odinger equation (\ref{dNLS}) in Lemma \ref{Global}, we obtain the result (\ref{Residu}).
\end{proof}

\section{Main Results}
\label{sec3}

We are now ready to formulate the main result of the paper that is stated in the following theorem: 
\begin{theorem}\label{theorem1}
	Let $u = (u_j)_{j\in \mathbb{Z}}$ be a solution of the dNLS equation \eqref{Parametric} and let $\phi =  (\phi_j)_{j\in \mathbb{Z}}$ be the leading approximation terms given by \eqref{ansatzP}. For every $\tau_0 > 0$, there are a small $\varepsilon_0 > 0$ and positive constants $C_0$ and $C$ such that for every $\varepsilon \in (0,\varepsilon_0)$ with 
	\begin{equation}\label{data1}
	\lVert u(0) - \phi(0) \rVert_{\ell^2(\mathbb{Z})} + \lVert \dot{u}(0) - \dot{\phi}(0)  \rVert_{\ell^2(\mathbb{Z})} \leq C_0 \varepsilon^3 \:,
	\end{equation}
	the inequality
	\begin{equation}\label{data2}
	\lVert u(t) - \phi(t) \rVert_{\ell^2(\mathbb{Z})} + \lVert \dot{u}(t) - \dot{\phi}(t)  \rVert_{\ell^2(\mathbb{Z})} \leq C \varepsilon^3 \:,
	\end{equation}
	holds for $t \in [0, 2\tau_0\varepsilon^{-2}]$.
\end{theorem}
\begin{proof}
Write
\begin{eqnarray}
u_j(t)= \phi_j(t) + y_j(t),\label{dec} 
\end{eqnarray}where $\phi_j(t)$ is the leading-order approximation (\ref{ansatzP}) and $y_j(t)$ is the error term. The error will give us a description of how good $\phi_j(t)$  is as an approximation to solutions of the dKG equation.

Plugging the decomposition \eqref{dec} into equation \eqref{Parametric}, we obtain the evolution problem for the error term as
\begin{equation}\label{EqError}
\ddot{y}_j + y_j + \xi \left(y_j^3 + 3 \phi_j^2 y_j + 3 \phi_j y_j^2\right) - \varepsilon^2 \Delta y_j + \epsilon^2 \hat{\alpha} \dot{y}_j - 2\varepsilon^2 h \cos(2\Omega t)\; y_j + R_j(t) = 0\:.
\end{equation}
Associated with equation \eqref{EqError}, we can define the energy of the error term as
\begin{equation}\label{energy}
	E(t) := \frac{1}{2} \sum_{j \in \mathbb{Z}} \left[\dot{y}_j^2 + y_j^2 - 2\varepsilon \left(y_j y_{j+1} - y_j^2\right)\right].
\end{equation}
Note that from the Cauchy-Schwartz inequality, we have
\begin{equation}
	\sum_{j \in \mathbb{Z}} y_j y_{j+1} \leq  \left(\sum_{j \in \mathbb{Z}} y_j^2\right)^{1/2}\left(\sum_{j \in \mathbb{Z}} y_{j+1}^2\right)^{1/2}= \lVert y \rVert_{\ell^2(\mathbb{Z})}^2 ,
\end{equation}
and 
\begin{equation}
 -2 \varepsilon \sum_{j \in \mathbb{Z}} y_j y_{j+1} + 2 \varepsilon \sum_{j \in \mathbb{Z}} y_j^2 \geq -2 \varepsilon \lVert y \rVert_{\ell^2(\mathbb{Z})}^2 + 2 \varepsilon \lVert y \rVert_{\ell^2(\mathbb{Z})}^2 = 0\:.
\end{equation}
Thus the energy is always positive for all $t$ on which the solution $y(t)$ is defined. We also have the inequality
\begin{equation}\label{normsol}
	\lVert \dot{y}(t) \rVert_{\ell^2(\mathbb{Z})}^2 + \lVert y(t) \rVert_{\ell^2(\mathbb{Z})}^2 \leq 2 E(t).
\end{equation}

From the energy \eqref{energy} and the error term \eqref{EqError}, we obtain that
\begin{eqnarray}\label{rate}
\dfrac{dE}{dt} &=& \dfrac{1}{2}\sum_{j \in \mathbb{Z}} \left[ 2 \ddot{y}_j \dot{y}_j + 2 \dot{y}_j y_j - 2 \varepsilon^2 \left(\dot{y}_j y_{j+1} + y_j \dot{y}_{j+1} - 2 y_j \dot{y}_j\right)\right] \nonumber \\
&=& \sum_{j \in \mathbb{Z}} \left[\ddot{y}_j + y_j - \varepsilon^2 (y_{j+1} + y_j - 2y_j)\right]\dot{y}_j \nonumber\\
&=& - \sum_{j \in \mathbb{Z}} \left[R_j(t) + \xi \left(y_j^3 + 3 \phi_j^2 y_j + 3 \phi_j y_j^2\right) + \epsilon^2 \hat{\alpha} \dot{y}_j - 2 \varepsilon^2 h \cos(2\Omega t)\;y_j\right] \dot{y}_j\:.
\end{eqnarray}
Setting $E = Q^2$ and using the Cauchy-Schwarz inequality, we have
\begin{eqnarray}\label{equivalent}
		\left|\frac{dQ}{dt}\right| &\leq& \frac{1}{\sqrt{2}}\lVert R(t)\rVert_{\ell^2(\mathbb{Z})} + \left[ |\xi| \left(2 Q^3 + 3 \lVert \phi \rVert_{\ell^2(\mathbb{Z})}^2 Q + 3\sqrt{2} \lVert \phi \rVert_{\ell^2(\mathbb{Z})} Q^2\right) +  \frac{\varepsilon^2\hat{\alpha}}{2} Q \right. \nonumber\\
		&& \qquad \left. + 2\varepsilon^2|h|Q^2 \right].
\end{eqnarray}

Take $\tau_0 > 0$ arbitrarily. Assume that the initial norm of the perturbation term satisfies the  bound
\begin{equation}\label{initialQ}
Q(0) \leq C_0 \varepsilon^3,
\end{equation}
where $C_0$ is a positive constant. Define 
\begin{eqnarray}\label{T0}
T_0 & = &\sup \left \{t_0 \in [0,2\tau_0\varepsilon^{-2}] : \sup_{t \in [0,t_0]} Q(t) \leq C_Q \varepsilon^3 \right \},\, C_R = \sup_{\tau \in [0,\tau_0]} \; \tilde{C_R},
\end{eqnarray}
on the time scale $[0, 2\tau_0\varepsilon^{-2}]$. Then, we can rewrite the energy estimate (\ref{equivalent}) by applying Lemmas \ref{ConstantX}--\ref{ConstantR} and the definition (\ref{T0}) as
\begin{equation}\label{estimateE}
		\left|\frac{dQ}{dt}\right| \leq \frac{1}{\sqrt{2}} C_R \varepsilon^5 + \left(4 |\xi| C_Q^2 \varepsilon^4 + 6 |\xi| C_\phi^2 + 6|\xi| \sqrt{2}  C_\phi C_Q \varepsilon^2 + \hat{\alpha} + 4 |h| \right) \frac{\varepsilon^2 Q}{2}\:.
\end{equation}
Thus, for every $t \in [0,T_0]$ and suffeciently small $\varepsilon > 0$, we can find a positive constant $K_0$, which is independent of $\varepsilon$, such that 
\begin{equation}
	4 |\xi| C_Q^2 \varepsilon^4 + 6 |\xi| C_\phi^2 + 6|\xi| \sqrt{2}  C_\phi C_Q \varepsilon^2 + \hat{\alpha} + 4 |h| \leq K_0 .
\end{equation}
By simplifying and integrating (\ref{estimateE}), we get
\begin{equation}
	Q(t) e^{-\frac{\varepsilon^2 K_{0} t}{2}} - Q(0) \leq \int_{0}^{t} \frac{C_R \varepsilon^5}{\sqrt{2}}  e^{-\frac{\varepsilon^2 K_0 s}{2}} ds \leq \frac{\sqrt{2}C_R \varepsilon^3}{K_0}\;.
\end{equation}
By using (\ref{initialQ}), then we obtain
\begin{equation}
	Q(t) \leq \varepsilon^3 \left(C_0 + \frac{\sqrt{2}C_R }{K_0}\right) e^{K_0 \tau_0}.
\end{equation}
Therefore, we can define $C_Q:= \left(C_0 + 2^{1/2} K_0^{-1} C_R\right) e^{K_0 \tau_0}$ and this 
concludes the proof.
\end{proof}

\section{Numerical simulations} 
\label{sec4}

\begin{figure}[tbhp!]
	\centering
	\subfloat[]{\includegraphics[scale=0.40]{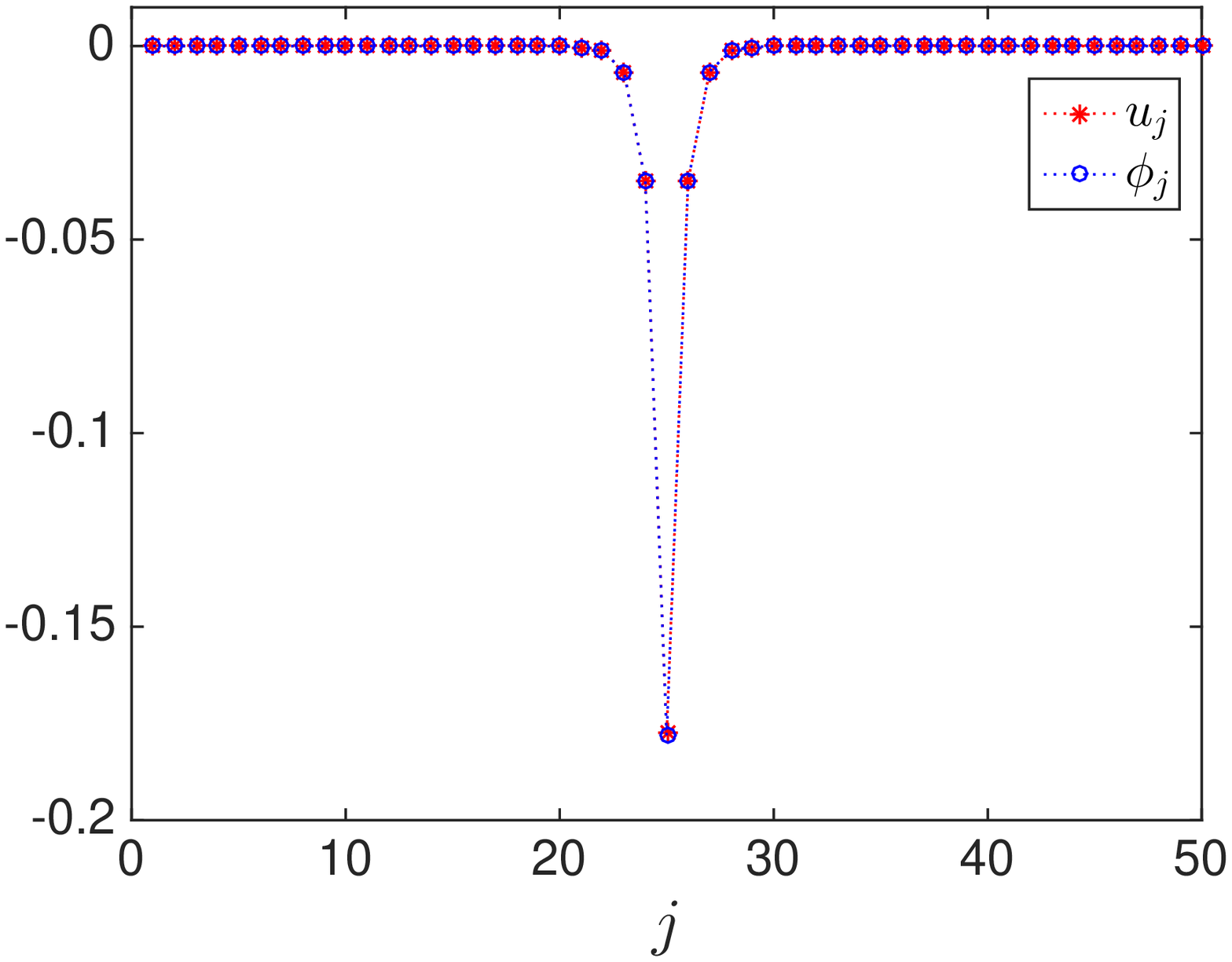}\label{fig1a}}
	\subfloat[]{\includegraphics[scale=0.40]{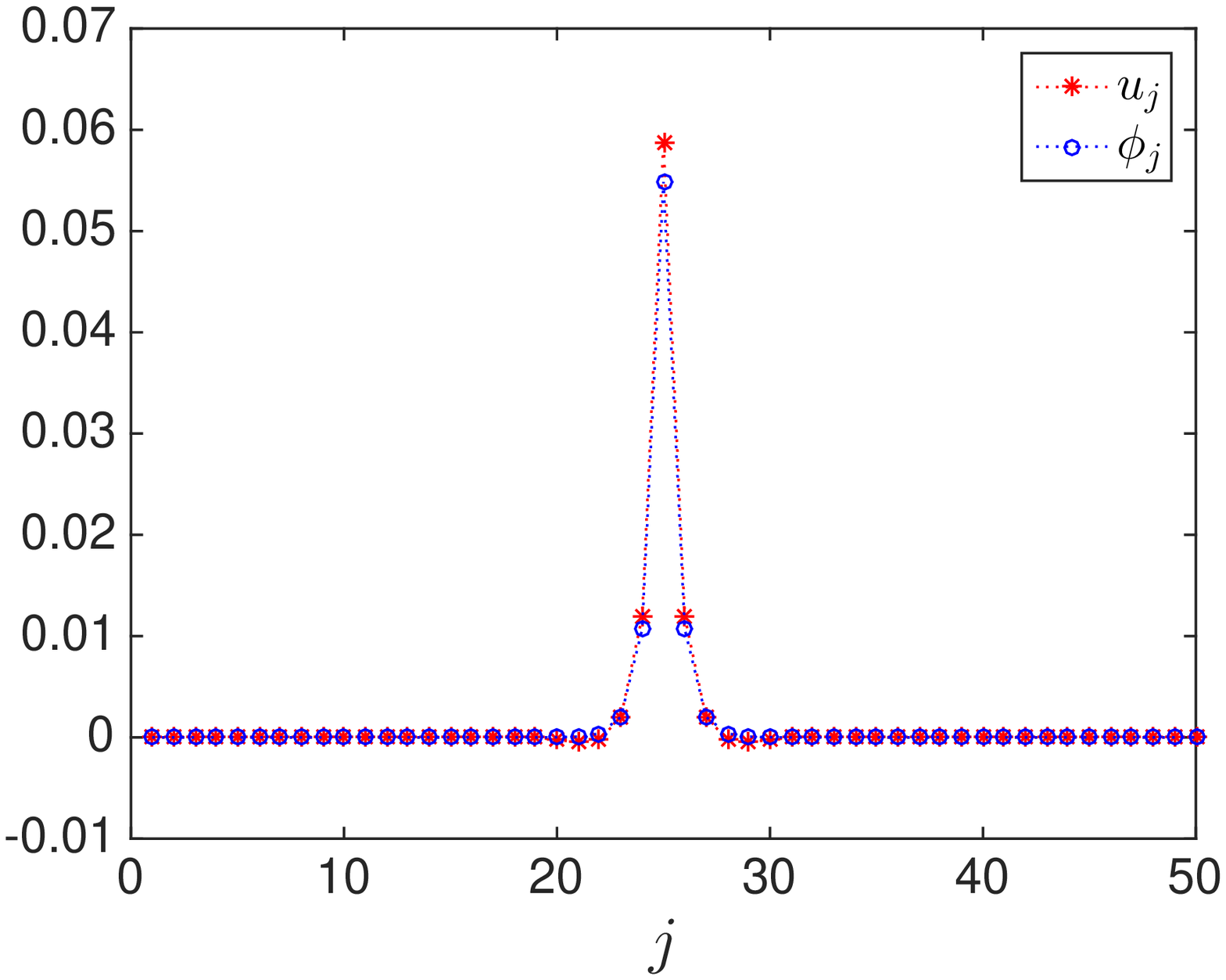}\label{fig1b}}\\
	\subfloat[]{\includegraphics[scale=0.40]{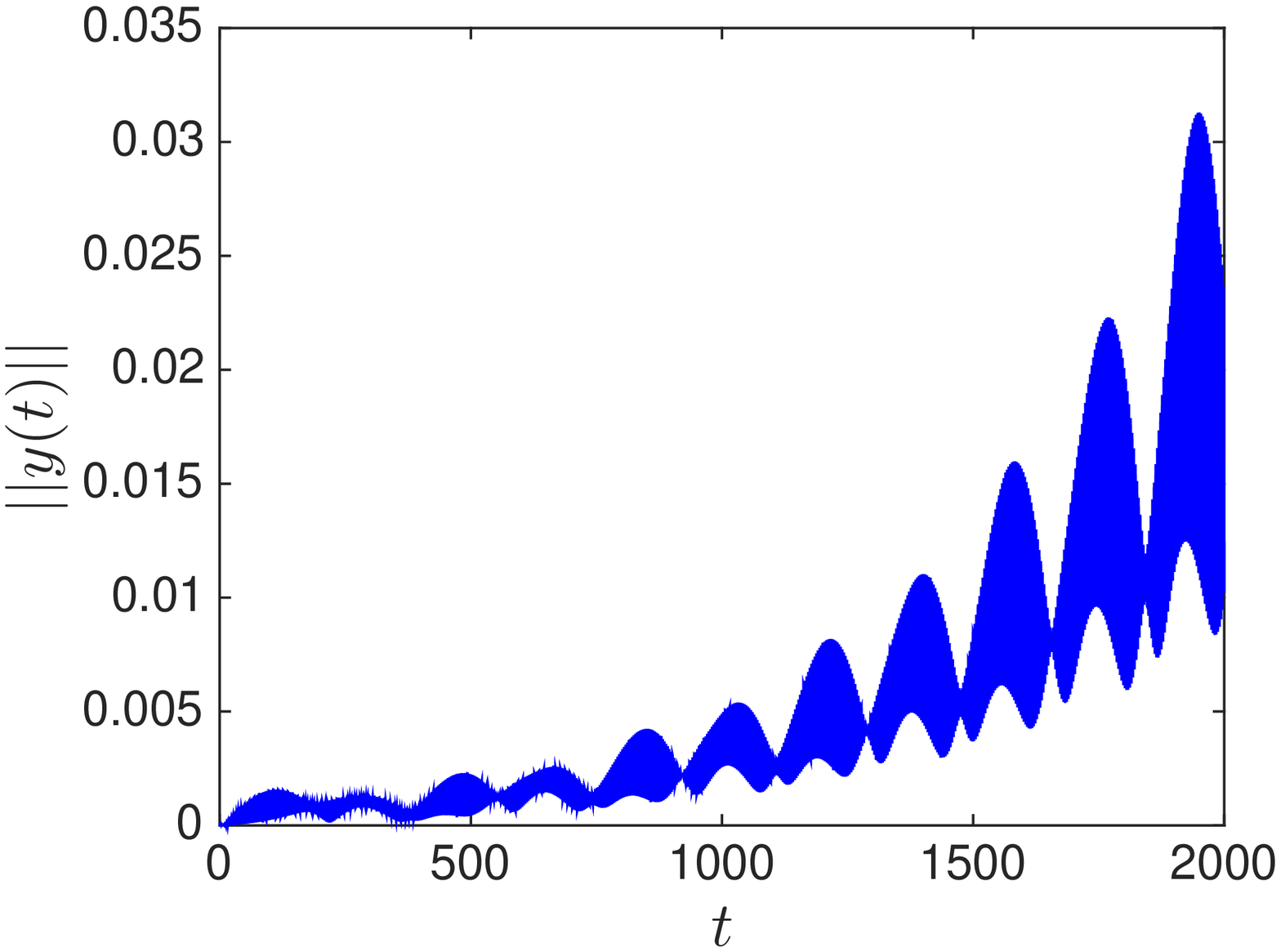}\label{fig1c}}
	\subfloat[]{\includegraphics[scale=0.40]{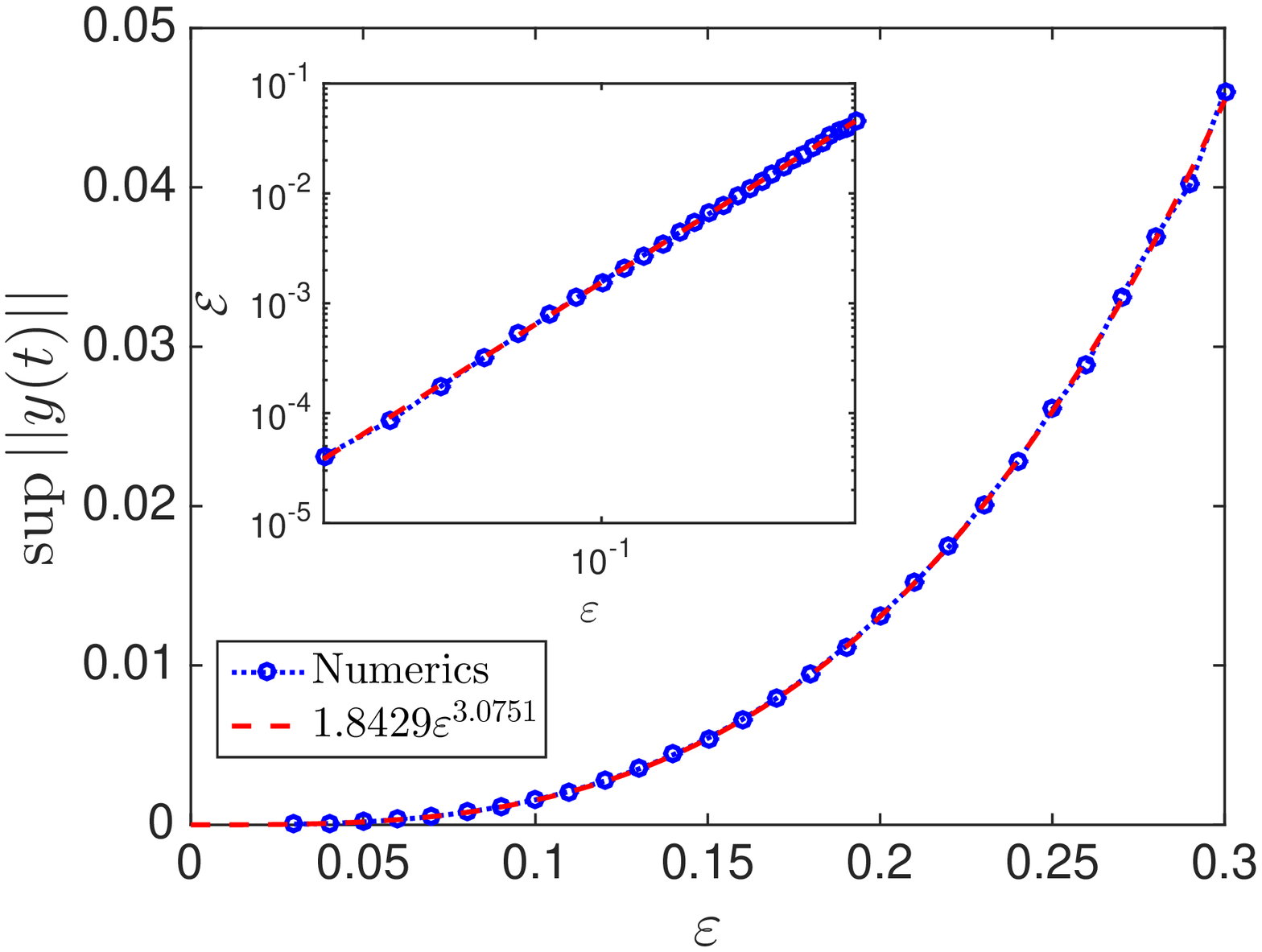}\label{fig1d}}
	\caption{{
			Panels (a,b) show numerical solutions of the dKG equation (blue circles) and the corresponding rotating wave approximations from the dNLS equation (red stars) at two time instances $t=100$ and $t=1000$. Here, $\varepsilon=0.1$. Panel (c) is the time dynamics of the error. Panel (d) is the maximum error of the dNLS approximation within the interval $t\in[0,2/\varepsilon^2]$ for varying $\varepsilon\to0$. In the picture, we also plot the best power fit of the error, showing the same order as in Theorem \ref{theorem1}.
		} 
	}
	\label{fig1}
\end{figure}

In Section \ref{sec:1}, we have discussed that small-amplitude solutions of the parametrically driven dKG equation \eqref{Parametric} can be approximated by ansatz \eqref{ansatzP}, that satisfies the dNLS equation \eqref{dNLS} with a residue of order $\mathcal{O}(\varepsilon^{5})$. We then showed in Section \ref{sec3} that the difference between solutions of Eqs.\ \eqref{Parametric} and \eqref{dNLS}, that are initially of at most order $\mathcal{O}(\varepsilon^3)$, will be of the same order for some finite time. In this section, we will illustrate the results numerically. 

We consider Eq.\ \eqref{Parametric} 
as an initial value problem in the domain $D = \{(n, t)|(n, t) \in [1, N] \times [0, {\widetilde{T}}]\}$, $N\in\mathbb{N},{\widetilde{T}} \in\mathbb{R}$. The differential equation is then integrated using the fourth order Runge-Kutta method. Simultaneously we also need to integrate Eq.\ \eqref{dNLS}. As the initial data of the dKG equation, we take 
\begin{eqnarray}
u_j(0)=\phi_j(0),\quad \dot{u}_j(0)=\left.\dot{\phi}_j(t)\right|_{t=0}.
\end{eqnarray}
In this way, the initial error $y(0)$ between $u_j(0)$ and $\phi_j(0)$ (see \eqref{dec}) will satisfy $\lVert y(0)\rVert_{\ell^2}=0<C_0\varepsilon^3$, for any $C_0>0$.

In the following, we take 
the parameter values $\Lambda=-3$, ${h}=-0.5$, and $\hat{\alpha}=0.1$. The nonlinearity is considered to be 'softening', which without loss of generality is taken to be  $\xi=-1$. This choice of nonlinearity coefficient will yield the dNLS equation \eqref{dNLS} with a 'focusing' or 'attractive' nonlinearity. The case $\xi=+1$, i.e., 'stiffening' nonlinearity, corresponds to the 'defocusing' or 'repulsive' dNLS equation \eqref{dNLS}. In the dNLS description, the attractive and repulsive cases are mathematically equivalent through a "staggering” transformation $(-1)^j$, that reverses the phases in every second lattice. 

In our first simulation, we consider the fundamental site-centred discrete soliton of the dNLS equation, 
that has been considered before in, e.g., \cite{hennig1999periodic,susanto2006stability,syafwan2010discrete,syaf12}. Such solutions will satisfy \eqref{dNLS} with $\dot{A}_j=0$ and can be obtained rather straightforwardly using Newton's method. 

In Fig.\ \ref{fig1a} and \ref{fig1b} we plot the solutions $u_j(t)$ and $\phi_j(t)$ for $\varepsilon=0.1$ at two different subsequent times. In panel (c) of the same figure, we plot the error $\lVert y(t)\rVert$ between the two solutions, which shows that it increases. However, the increment is bounded within the proven estimate $\sim C\epsilon^3$ for quite a long while.

We have performed similar computations for several different values of $\varepsilon\to0$. Taking $\tau_0=1$, we record sup$_{t\in[0,2\tau_0/\varepsilon^2]}\lVert y(t)\rVert$ for each $\varepsilon$. We plot in Fig.\ \ref{fig1d} the maximum error as a function of $\varepsilon$. We also plot in the same panel the best power fit in the nonlinear least squares sense
, which agrees with Theorem \ref{theorem1}.

Discrete solitons of the dNLS equation expectedly approximate discrete breathers of the dKG equation. Our simulations above indicate this as well. Yet, how close are they with each other? In \ref{tambh}, we show numerically that they are $\mathcal{O}(\varepsilon^3)$-apart, which interestingly seem to follow the result in Theorem \ref{theorem1}.

\begin{figure}[bt]
	\centering
	\subfloat[]{\includegraphics[scale=0.40]{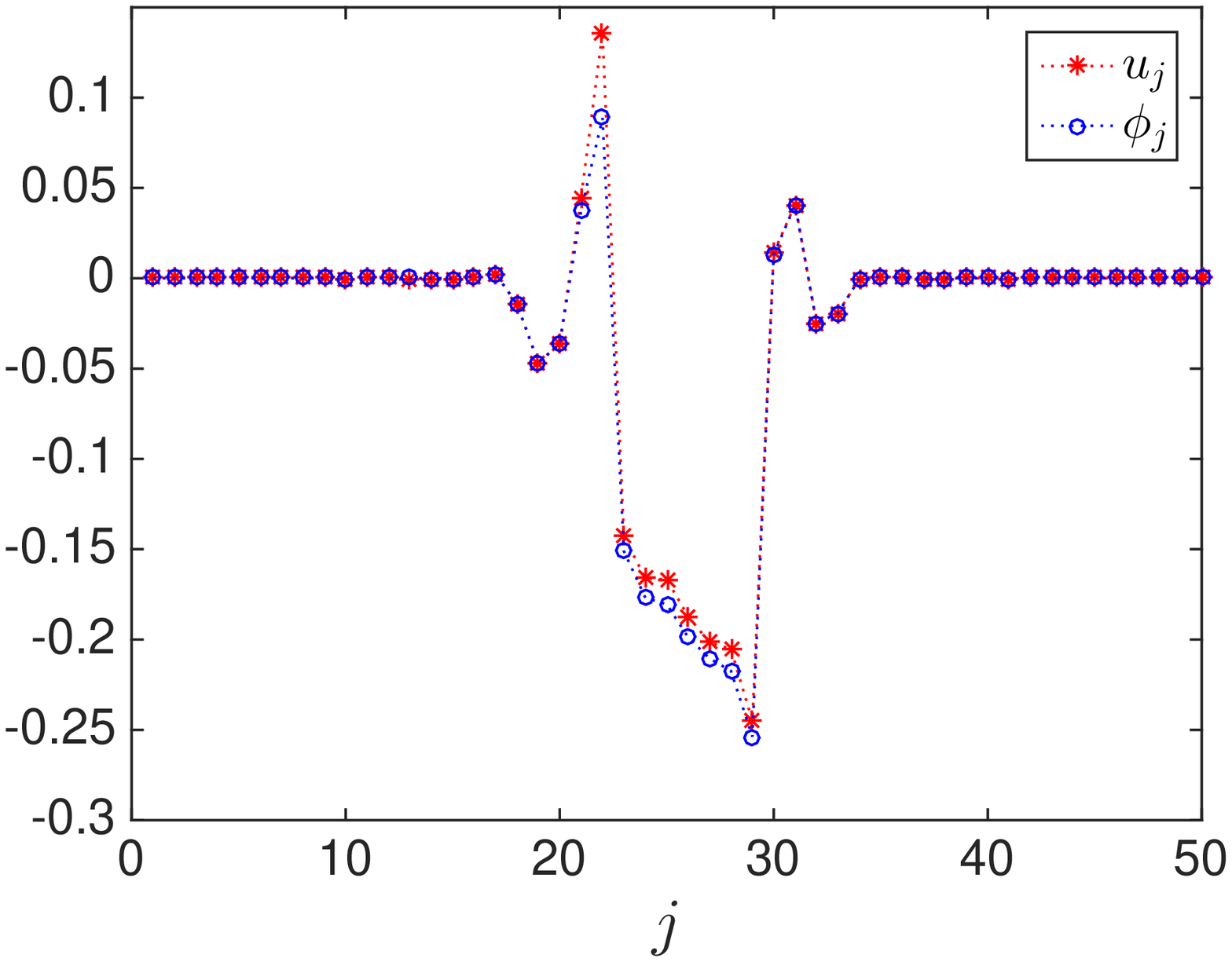}\label{fig2a}}
	\subfloat[]{\includegraphics[scale=0.40]{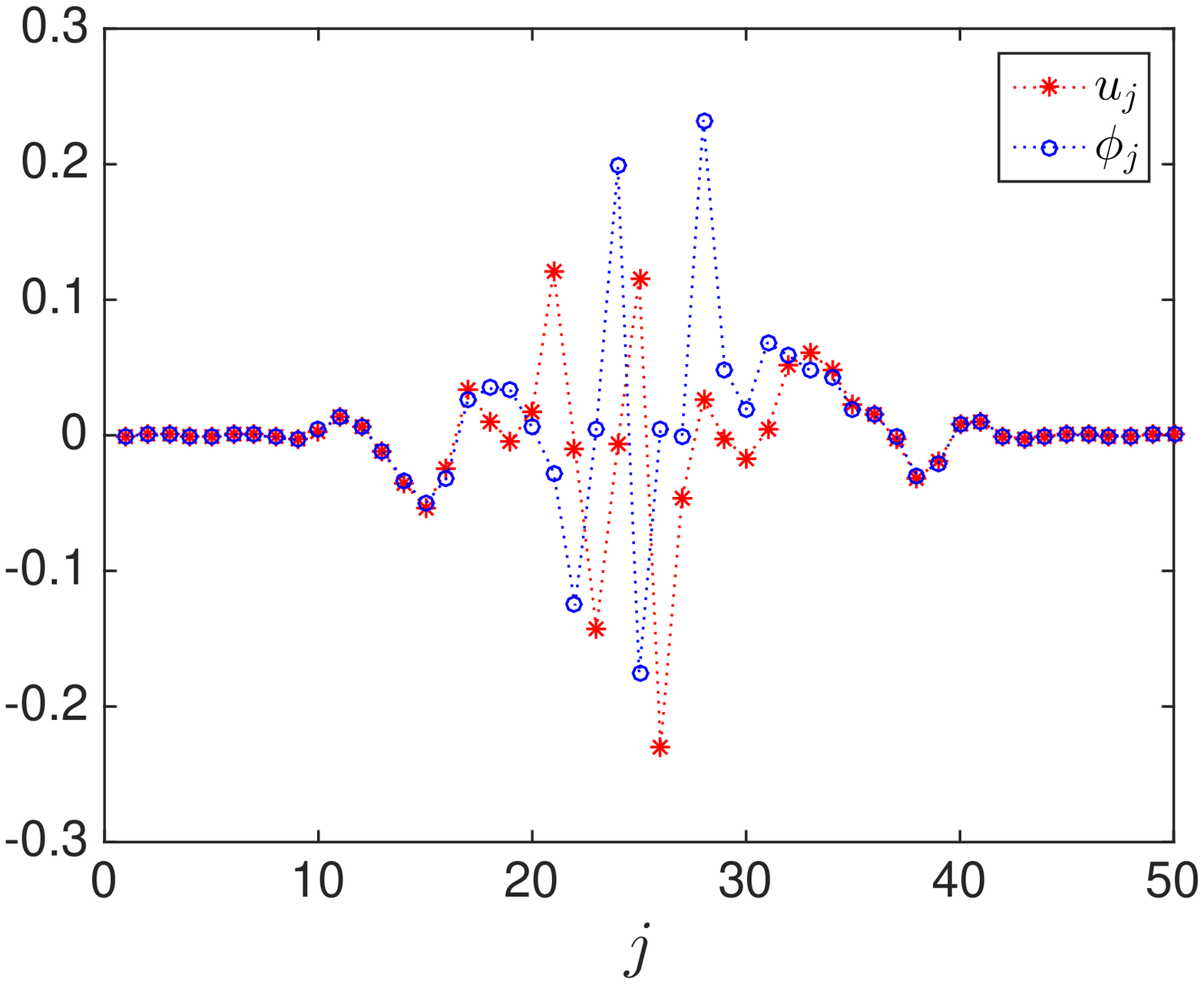}\label{fig2b}}\\
	\subfloat[]{\includegraphics[scale=0.40]{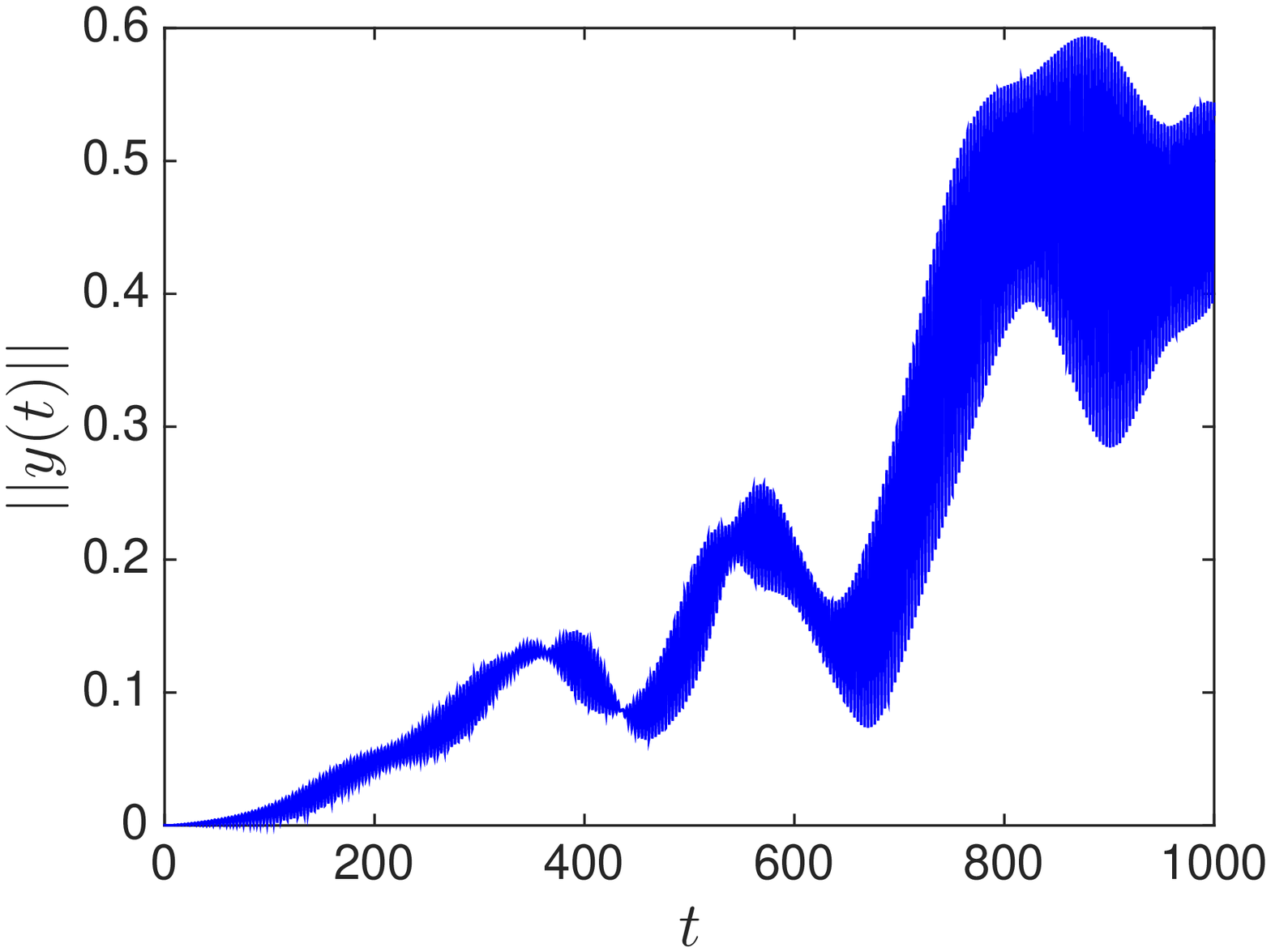}\label{fig2c}}
	\subfloat[]{\includegraphics[scale=0.40]{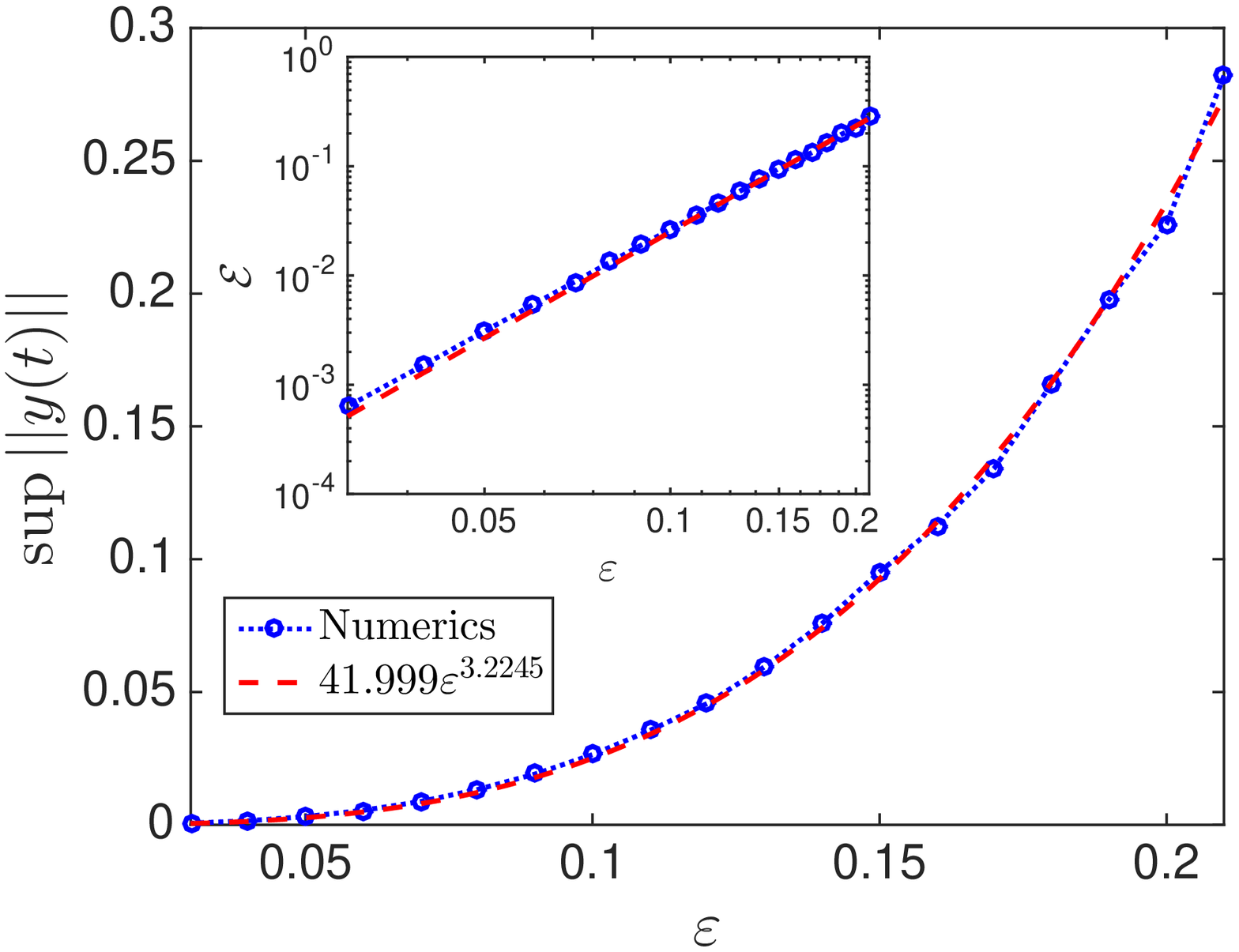}\label{fig2d}}
	\caption{{
			The same as Fig.\ \ref{fig1}, but for the initial data \eqref{in2}.
		} 
	}
	\label{fig2}
\end{figure}

In our second simulations, we consider a perhaps more interesting initial condition in the form of a clustered state:
\begin{eqnarray}
A_j=e^{0.05ij},\,j=21,\dots,30,
\label{in2}
\end{eqnarray}
and $A_j$ vanishes elsewhere. The dynamics at some instances are shown in Fig.\ \ref{fig2}. We also computed the maximum error made by the rotating wave approximation within the time interval $[0,2\tau_0/\varepsilon^2]$, with $\tau_0$ taken to be 1, and plotted it in Fig.\ \ref{fig2d} for several values of $\varepsilon$. The best power fit to the error also shows the same behaviour, i.e., the error is $\mathcal{O}(\varepsilon^3)$.

\section*{Acknowledgement}
YM thanks MoRA (Ministry of Religious Affairs) Scholarship of the Republic of Indonesia for a financial support. The research of FTA and BEG are supported by PDUPT Kemenristekdikti 2018. RK gratefully acknowledges financial support from Lembaga Pengelolaan Dana Pendidikan (Indonesia Endowment Fund for Education) (Grant No.\ - Ref: S-34/LPDP.3/2017). The authors are grateful to the four reviewers for their comments that improved the quality of the manuscript. 

\section*{Compliance with Ethical Standards}	
\noindent\textbf{Conflicts of interest}: The authors declare that they have no conflict of interest.

\appendix
\setcounter{subfigure}{0} 
\section{Discrete breathers vs.\ discrete solitons}
\label{tambh}

\begin{figure}[tbhp]
	\centering
	\subfloat[]{\includegraphics[scale=0.40]{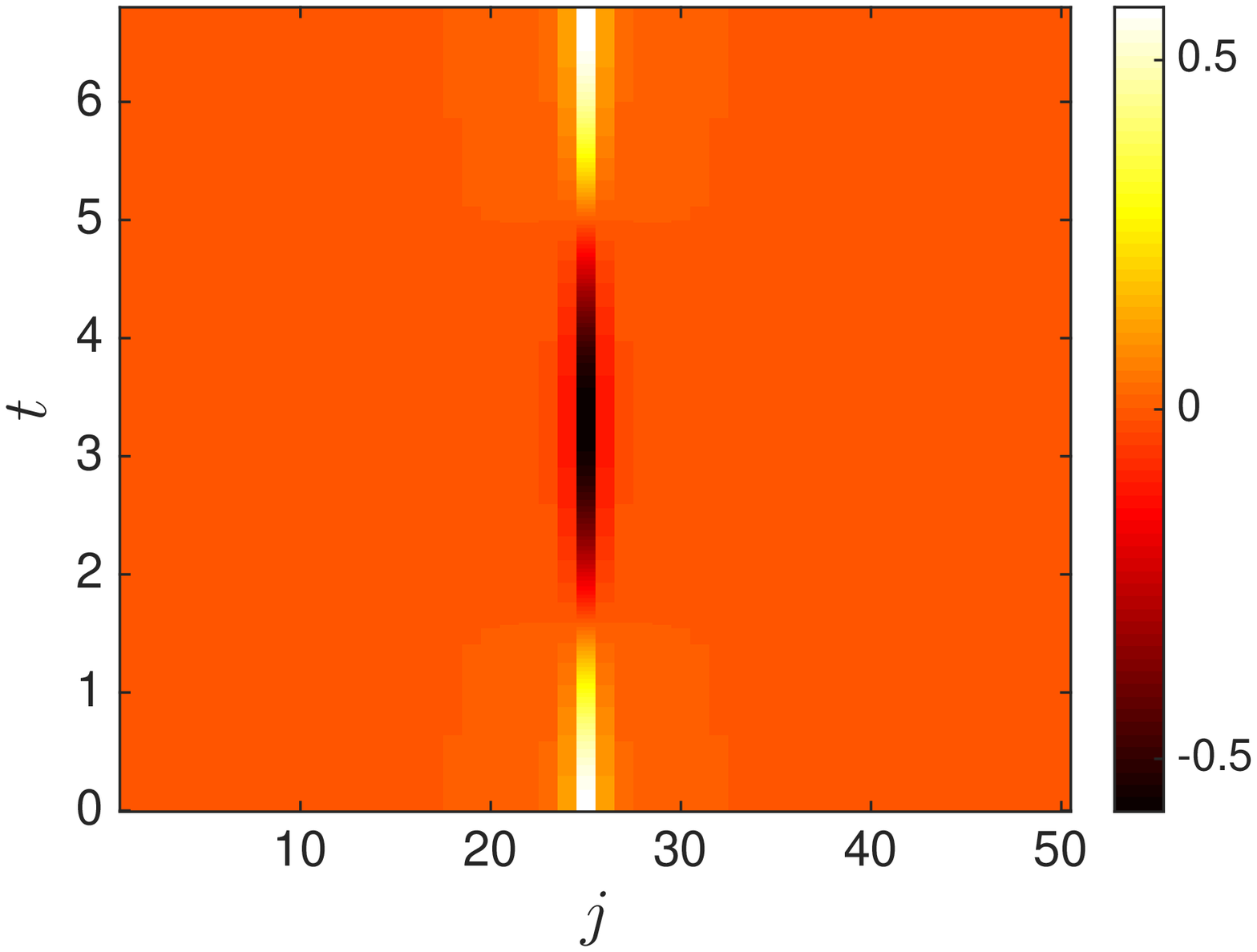}\label{subfig:prof_kg_eps_0_05}}
	\subfloat[]{\includegraphics[scale=0.40]{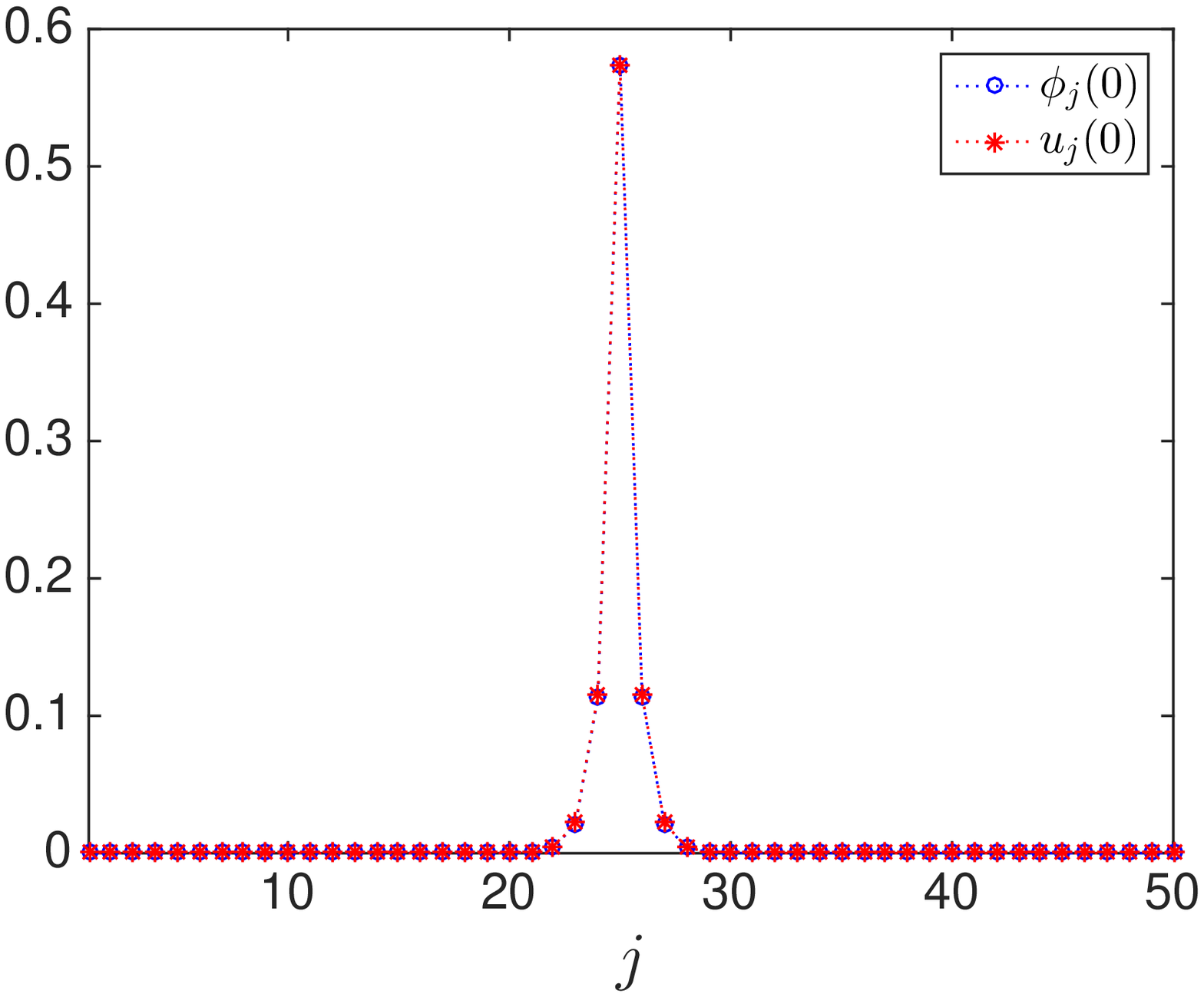}\label{subfig:prof_compare_kg_dnls_eps_0_05}}
	\caption{A breather of \eqref{Parametric} for ${\varepsilon^2}=0.05$.
		Panel (a) shows the dynamics of the solution in one period, while panel (b) presents the comparison of the breather and its approximation \eqref{ansatzP} at $t=0$, with $A_j$ being a discrete soliton of Eq.\ \eqref{dNLS}.
	}
	\label{fig:onsite_stability_compare_kg_dnls}
\end{figure}

While discrete solitons of the dNLS  \eqref{dNLS} correspond to spatially localised, but time-independent solutions of the equation and can be computed rather immediately, discrete breathers are spatially localised, but temporally periodic solutions of the dKG equation \eqref{Parametric}.

There are several numerical methods that have been developed to seek for discrete breathers, see the review \cite{flac04,flac08}. Here, we use a Fourier series representation by writing $u_j(t)$ as
\begin{equation}
u_j(t)=\sum_{k=1}^{K}a_{j,k}\cos\left(\left(k-1\right)\Omega t\right)+b_{j,k}\sin\left(k\Omega t\right),
\label{eq:anz}
\end{equation}
where $a_{j,k}$ and $b_{j,k}$ are the Fourier coefficients and $K\gg1$ is the number of Fourier modes used in our numerics. By substituting the expansion \eqref{eq:anz} into the dKG equation \eqref{Parametric}, multiplying with each mode, and integrating it over the time-period $2\pi/\Omega$, one will obtain coupled nonlinear algebraic equations for the coefficients $a_{j,k}$ and $b_{j,k}$.
Then, we use a Newton's method to solve the equations. Breather solutions will be obtained by choosing a proper initial guess for the coefficients $a_{j,k}$ and $b_{j,k}$.

\begin{figure}[tbhp]
	\centering
	{\includegraphics[scale=0.50]{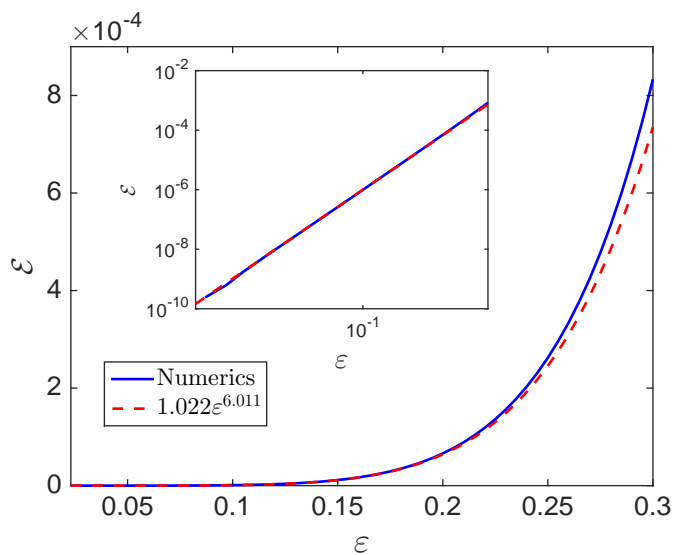}}
	\caption{Plot of the maximum difference \eqref{erro} of the discrete Schr\"odinger approximation (\ref{dNLS}) for varying $\varepsilon$. The dashed line is the best power fit, indicated in the legend. The inset shows the curves in a log scale. 
	}
	\label{subfig:error_plot}
\end{figure}

Once a discrete soliton or a discrete breather is found, it is naturally relevant to study their stability. 

Let $\tilde{A}_j=\hat{x}_j + i\hat{y}_j$ be a discrete soliton of the dNLS equation. We determine its linear stability by writing 
\begin{equation}
A_j=\tilde{A}_j+\delta(\hat{x}_j + i\hat{y}_j)e^{\lambda {\tau}}.
\label{eq:anz_stab}
\end{equation}
Substituting \eqref{eq:anz_stab} into \eqref{dNLS} and linearising around $\delta=0$ will yield the eigenvalue problem
\begin{equation}
\lambda\left(
\begin{array}{c}
\hat{x}_j\\
\hat{y}_j
\end{array}
\right)=
\left(
\begin{array}{cc}
-6\xi{x}_j{y}_j-\alpha& \Delta+\Lambda-h-3\xi\left({x}_j^2+3{y}_j^2\right)\\
-\Delta-\Lambda-h+3\xi\left(3{x}_j^2+{y}_j^2\right)&6\xi{x}_j{y}_j-\alpha
\end{array}
\right)
\left(
\begin{array}{c}
\hat{x}_j\\
\hat{y}_j
\end{array}
\right).
\label{em2}
\end{equation}
In here, the solution $\tilde{A}_j$ is said to be linearly stable when all of the eigenvalues $\lambda$ have Re$(\lambda)\leq0$ and unstable when there is an eigenvalue with Re$(\lambda)>0$.

As for discrete breathers of the dKG equation, their linear stability is determined using Floquet theory that can be computed numerically as follows. Let $\tilde{u}_j(t)$ be a breather solution. By defining $u_j(t) = \hat{u}_j(t)+\delta Y_j(t)$, substituting it into Eq.\ \eqref{Parametric}, and linearising the equation around $\delta=0$, we obtain the system of linear differential-difference equations 
\begin{equation}
\begin{array}{ccl}
\dot{Y}_j&=&Z_j\\
\dot{Z}_j&=&-Y_j-3\xi \hat{u}_j^2 Y_j+\varepsilon^2\Delta_2Y_j-\alpha {Z}_j+H\cos\left(2\Omega t\right)Y_j.
\end{array}
\label{fm}
\end{equation}
Integrating Eqs.\ \eqref{fm} in the numerical domain $D$, where now $\widetilde{T}={2\pi}/{\Omega}$, and using the standard basis in $\mathbb{R}^{2N}$, i.e., $\left\{e^0_1,e^0_2,...,e^0_{2N}\right\}$ as the initial condition at $t=0$, we will obtain a set of solutions at $t=\widetilde{T}$, which is our monodromy matrix
\begin{equation}
M=\left\{
E^{\widetilde{T}}_1,E^{\widetilde{T}}_2,...,E^{\widetilde{T}}_{2N}
\right\}
\in \mathbb{R}^{2N\times2N}.
\label{mm}
\end{equation}
The breather $\tilde{u}_j(t)$ is said to be linearly stable when all the eigenvalues $\lambda_\text{dKG}$ of the monodromy matrix $M$, which are known as Floquet multipliers, lie inside or on the unit circle and unstable when there exists at least one $\lambda_{dKG}$ lying outside the unit circle. Note that in the presence of damping, the set of continuous multipliers will lie on a circle of radius $e^{\frac{-\alpha \pi}{\Omega}}$, see \cite{mari01,mart03}.

\begin{figure}[tbhp!]
	\centering
		\subfloat[]{\includegraphics[scale=0.40]{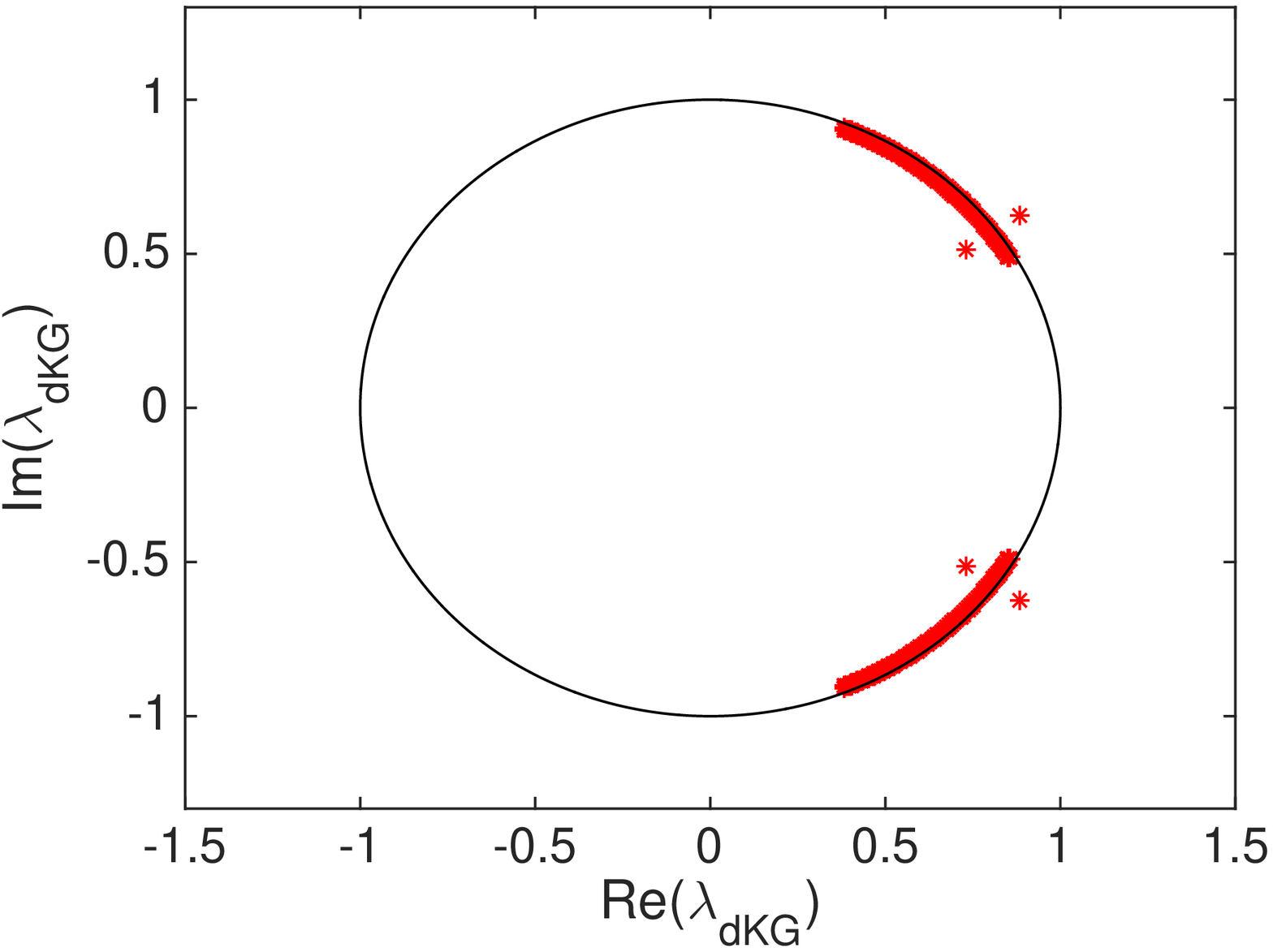}\label{subfig:eig_kg_eps_0_05}}
	\subfloat[]{\includegraphics[scale=0.40]{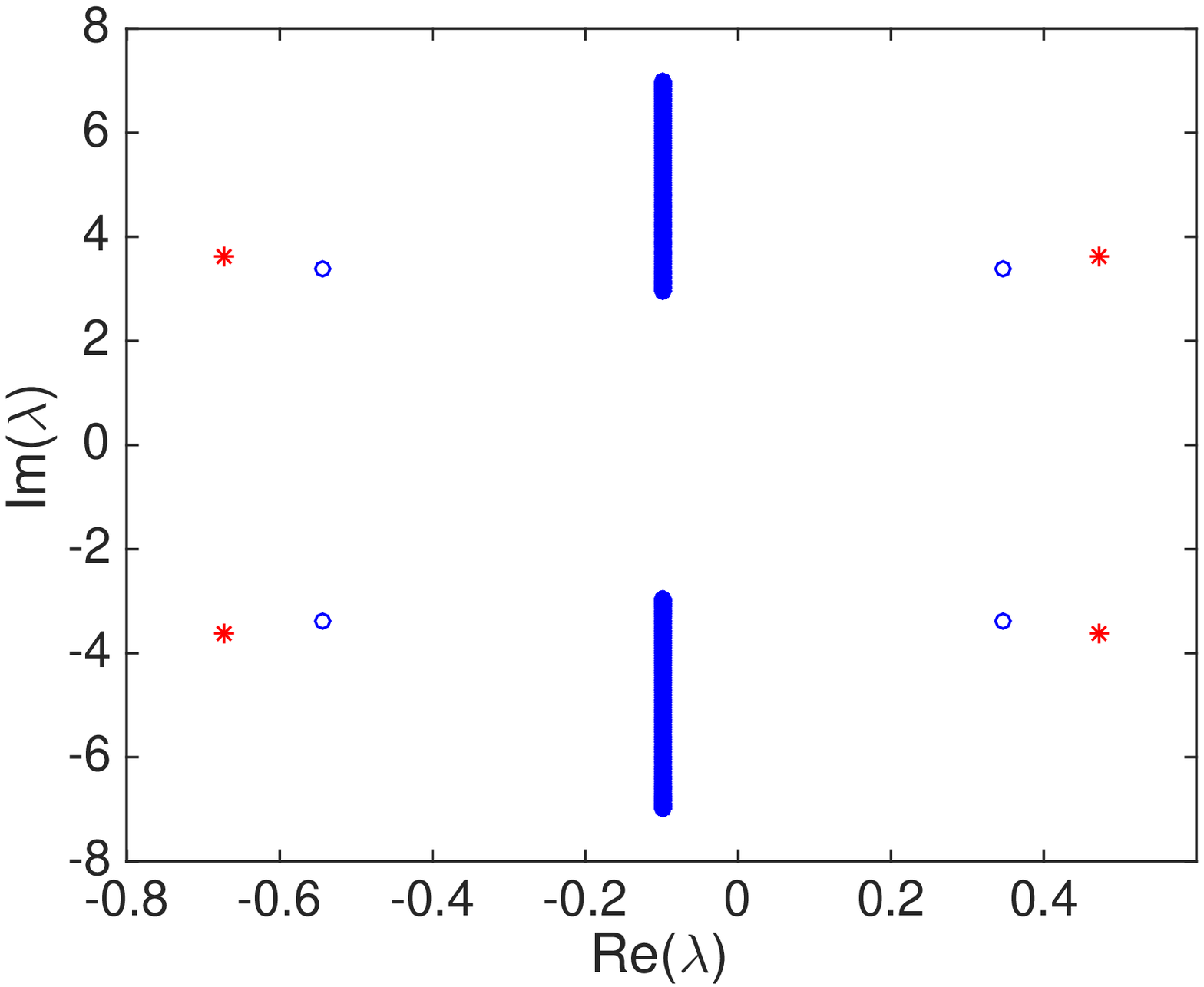}\label{subfig:eig_kg_dnls_eps2_0_05}}\\
	\subfloat[]{\includegraphics[scale=0.40]{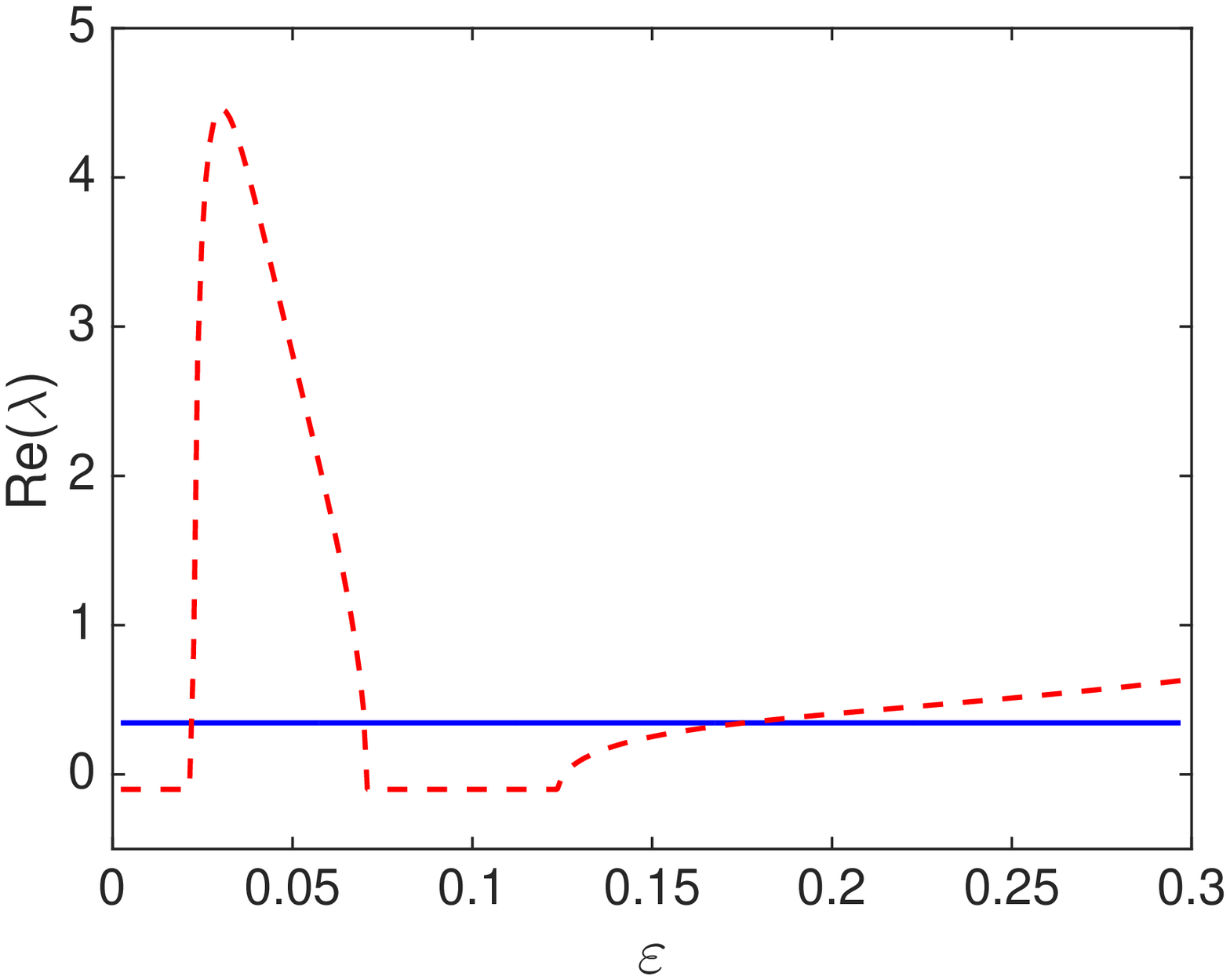}\label{subfig:Re_eig_kg_dnls_vary_eps}}
\subfloat[]{\includegraphics[scale=0.40]{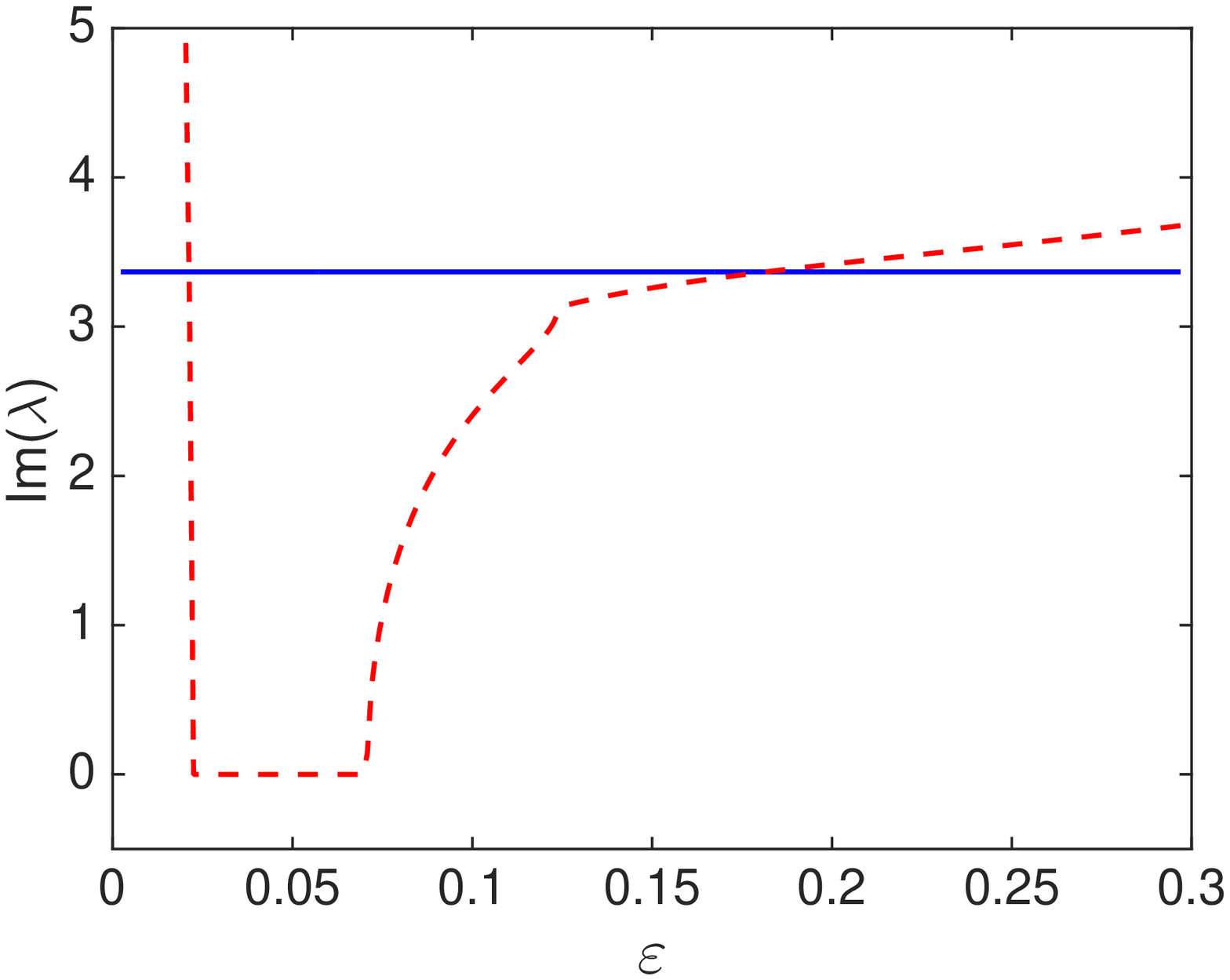}\label{subfig:Im_eig_kg_dnls_vary_eps}}
\caption{{
			Panel (a) shows Floquet multipliers of the breather in Fig.\ \ref{subfig:prof_kg_eps_0_05}, showing the linear instability of the solution. Panel (b) presents the eigenvalues of the corresponding discrete soliton of DNLS equation \eqref{dNLS}. Red stars in the panel are the critical multipliers in panel (a), that have been transformed following the relation \eqref{tran}. 
		Panels (c) and (d) compare the real and imaginary part of the critical eigenvalue of the discrete soliton (blue solid line) and the critical multiplier of the corresponding breather of the dKG equation (red dashed line) for varying $\varepsilon$. 
	} 
	}
	\label{fig:onsite_stability_test}
\end{figure}

\begin{figure}[tbph]
	\centering
	\subfloat[]{\includegraphics[scale=0.40]{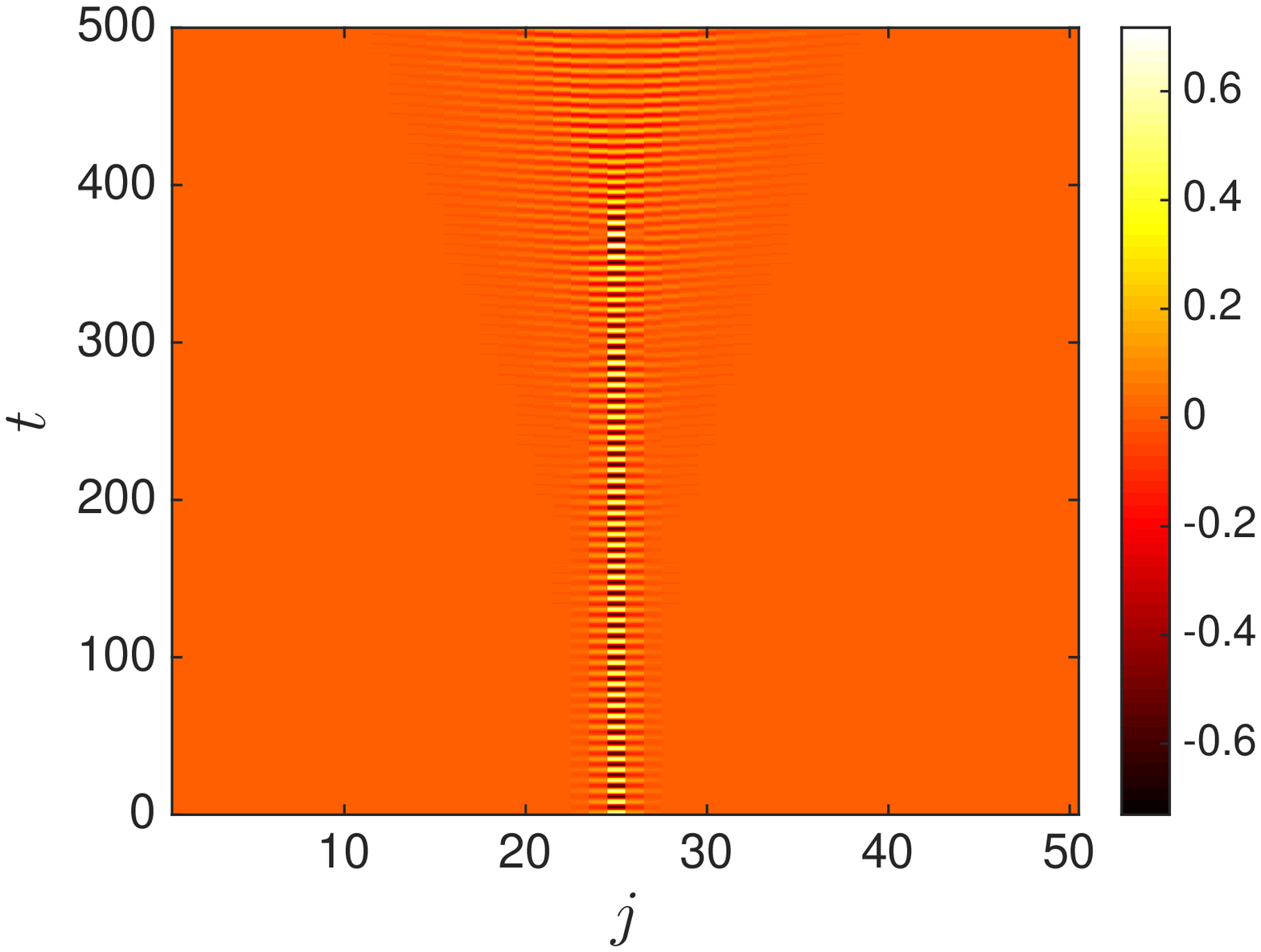}\label{subfig:KG_dynamics_eps2_0_05}}
	\subfloat[]{\includegraphics[scale=0.40]{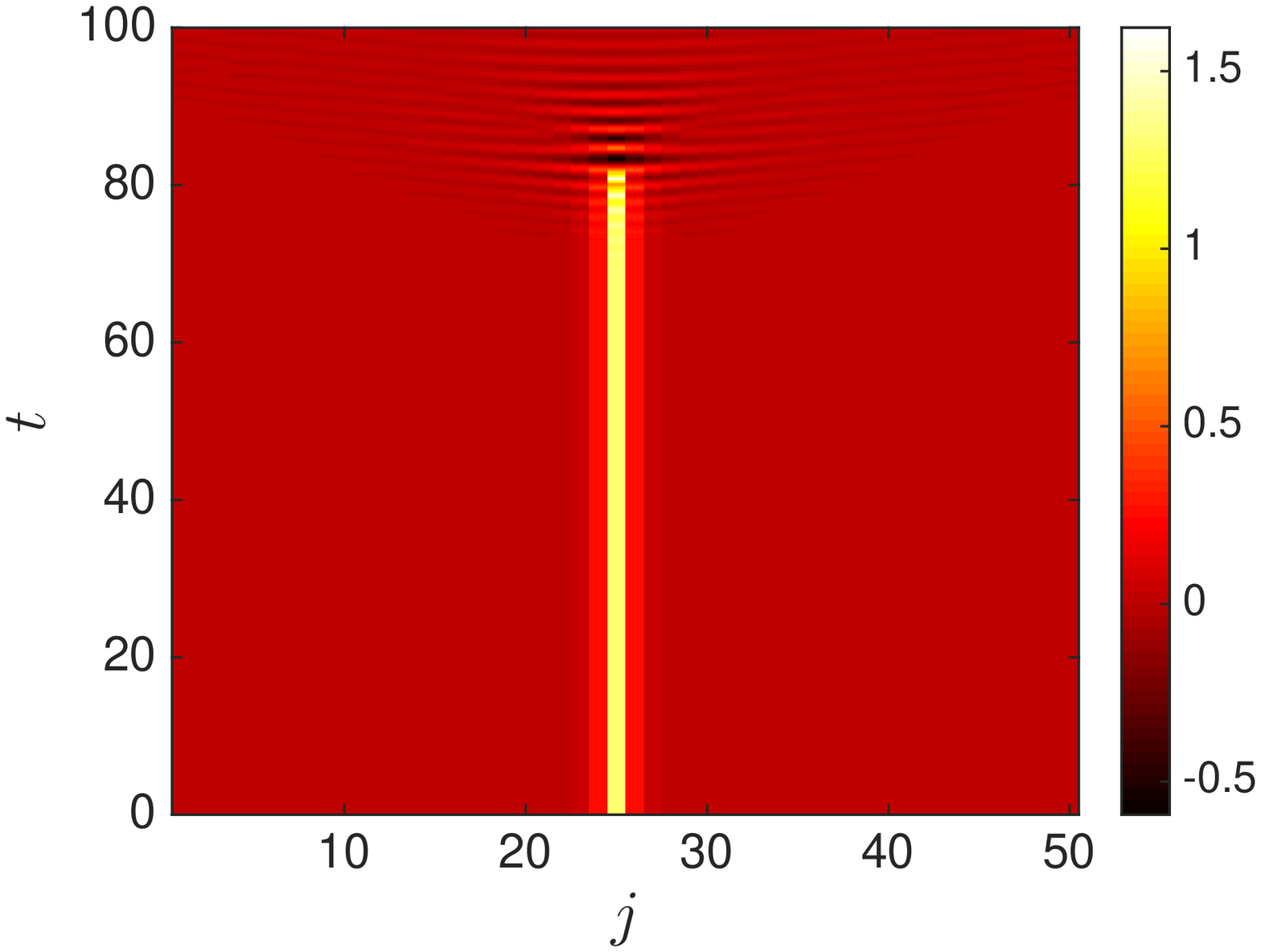}\label{subfig:DNLS_dynamics_eps2_0_05}}
	\caption{
{
					Time dynamics of the unstable breather (a) and discrete soliton (b) shown in Fig.\ \ref{fig:onsite_stability_compare_kg_dnls}. Note that the time variable in the second panel has been scaled to the original one. 
			}
	}
	\label{fig:eig_kg_dnls_vary_eps}
\end{figure}

In the following, we focus on breathers and discrete solitons for the same parameter values as in Section \ref{sec4}, i.e.,  $\Lambda=-3$, ${h}=-0.5$, $\hat{\alpha}=0.1$, and $\xi=-1$. For $\xi=+1$, due to the staggering transformation explained briefly in Section \ref{sec4}, discrete breathers of \eqref{Parametric} with small amplitudes and discrete solitons of \eqref{dNLS} will have exponentially decaying staggered tails.

For our computations, we solve the dKG equation \eqref{Parametric} for periodic in time solutions using the number of Fourier modes $K=3$ and the lattice sites $N=50$. Larger numbers, i.e., $K=9$ and $N=400$, have been used as well to make sure that the results are independent of the lattice size and the number of modes.

We present a breather solution and its time dynamics within one period in Fig.\ \ref{fig:onsite_stability_compare_kg_dnls} for $\varepsilon^2={0.05}$.
In Fig.\ \ref{subfig:prof_compare_kg_dnls_eps_0_05}, we compare the breather in panel (a) with its corresponding approximation \eqref{ansatzP}, where $A_j$ is the discrete soliton of Eq.\ \eqref{dNLS}.
One can note that they are in good agreement.

By defining a maximum difference between breathers of \eqref{Parametric} and their approximations \eqref{ansatzP} using discrete solitons of \eqref{dNLS} as
\begin{equation}
\mathcal{E}=\sup_{0\leq t\leq2\pi/\Omega}\lVert y(t) \rVert_{\ell^2},
\label{erro}
\end{equation}
we depict the error for varying $\varepsilon$ in Fig.\ \ref{subfig:error_plot}}. We also present in the same panel, the best power fit to the numerical results, which interestingly follows the theoretical prediction of the error in Theorem \ref{theorem1}, i.e., $\sim\varepsilon^{3}$.

We have computed the corresponding monodromy matrix for the stability of the solution in Fig.\ \ref{subfig:prof_kg_eps_0_05}. The Floquet multipliers are plotted in Fig.\ \ref{subfig:eig_kg_eps_0_05}. We have also solved Eq.\ \eqref{em2} for the corresponding discrete soliton of the dNLS equation \eqref{dNLS} and plot the eigenvalues $\lambda$ in Fig.\ \ref{subfig:eig_kg_dnls_eps2_0_05}, where interestingly we obtain that both solutions experience the same type of instability (i.e., oscillatory instability as the critical multipliers and eigenvalues are both complex valued). For the dNLS solitons, this is in agreement with the results of Refs.\ \cite{susanto2006stability}. For the dKG breather, the instability is similar to that reported in \cite{cuev09}. Moreover, from the time scales that lead to Eqs.\ \eqref{ansatzP}--\eqref{dNLS}, we can obtain the relation between Floquet multipliers $\lambda_\text{dKG}$ of the dKG monodromy matrix \eqref{mm} and eigenvalues $\lambda$ of the dNLS stability matrix \eqref{em2}, i.e., 
\begin{equation}
\displaystyle\lambda_\text{dKG}\sim e^{\pi\varepsilon^2\lambda/\Omega}.
\label{tran}
\end{equation}
Using the transformation, we depict in Fig.\ \ref{subfig:eig_kg_dnls_eps2_0_05} the critical multipliers as red stars, where we learn that the localised solutions do not only have the same type of instability, but their critical eigenvalues have relatively comparable magnitudes. 

While in Fig.\ \ref{subfig:error_plot} we plot the error made by the dNLS solitons in approximating the dKG breathers, in Figs.\ \ref{subfig:Re_eig_kg_dnls_vary_eps} and \ref{subfig:Im_eig_kg_dnls_vary_eps} we compare their critical eigenvalues and multipliers for varying $\varepsilon$. We obtain that the breathers and solitons do not necessarily share the same type of stability. In fact, there are intervals of coupling constant $\varepsilon$ on which the breathers are stable, even though the corresponding dNLS solitons are unstable. This observation nonetheless does not violate the analytical results in Sections \ref{sec:1}--\ref{sec3}.


Figures \ref{subfig:KG_dynamics_eps2_0_05} and \ref{subfig:DNLS_dynamics_eps2_0_05} show the typical dynamics of the oscillatory instability of the breather and its corresponding discrete soliton approximation. We can observe that in both cases the instability destroys localised solutions, i.e., we also obtain qualitative agreement in the instability dynamics. 


Finally, for the sake of completeness, we studied the typical dynamics of dKG breathers when they experience an exponential instability, i.e., the critical Floquet multipliers are real. In Figs.\ \ref{subfig:Re_eig_kg_dnls_vary_eps} and \ref{subfig:Im_eig_kg_dnls_vary_eps}, they are in a finite interval close to the uncoupled limit $\varepsilon=0$ and their absolute magnitudes are near unity. Due to these facts, we could not clearly see any instability in our simulations, even after integrating the dKG equation for quite a long while. 

\end{document}